\documentclass[a4paper,twocolumn,accepted=2020-06-09]{quantumarticle}
\pdfoutput=1
\usepackage[utf8]{inputenc}
\usepackage[english]{babel}
\usepackage[T1]{fontenc}
\usepackage{amsmath}
\usepackage[breaklinks]{hyperref}

\usepackage{tikz}

\usepackage{graphicx}
\usepackage{amsthm}
\usepackage{amssymb}
\usepackage{bm}         % bold math
\usepackage{mathtools}  % extra mathtools, some things fixes AMS math bugs
\usepackage{dsfont}	% Double stroke roman fonts, e.g. \mathds{1}
\usepackage{dcolumn} 	% Align table columns on decimal point
\usepackage{placeins}   % FloatBarrier command
\usepackage{centernot}  % for centered negations (cross out slashes)7
\usepackage{microtype}  % Subliminal refinements towards typographical perfection. 
\usepackage{xcolor}
\usepackage[numbers,sort&compress]{natbib}

\hypersetup{
    colorlinks,
    linkcolor={red!50!black},
    citecolor={blue!50!black},
    urlcolor={blue!80!black},
    pdfpagemode=UseNone, % don't show bookmarks when opening pdf file
    bookmarksopen=false,
    pdfstartview=FitH,
    pdfview=FitH
}

\newtheorem{theorem}{Theorem}
\newtheorem{proposition}[theorem]{Proposition}
\newtheorem{observation}[theorem]{Observation}
\newtheorem{lemma}[theorem]{Lemma}
\newtheorem{corollary}[theorem]{Corollary}
\newtheorem{conjecture}[theorem]{Conjecture}

\theoremstyle{remark}
\newtheorem*{remark}{Remark}

\newtheorem*{researchproblem}{Research Problem}

\theoremstyle{definition}

\DeclareMathOperator{\tr}{tr}

\DeclareMathOperator{\wt}{wt}
\DeclareMathOperator{\supp}{supp}

\newcommand{\ot}[0]{\otimes}
\newcommand{\nn}[0]{\nonumber}
\newcommand{\nhf}[0]{\lfloor \frac{n}{2} \rfloor}

\newcommand{\q}[0]{\quad}

\renewcommand{\AA}[0]{\mathcal{A}}
\newcommand{\BB}[0]{\mathcal{B}}
\newcommand{\CC}[0]{\mathcal{C}}

\newcommand{\HH}[0]{\mathcal{H}}

\newcommand{\QQ}[0]{\mathcal{Q}}

\newcommand{\UU}{\mathcal{U}}

\newcommand{\R}{\mathds{R}}
\newcommand{\C}{\mathds{C}}
\newcommand{\F}{\mathds{F}}

\newcommand{\one}[0]{\mathds{1}}

\newcommand{\bra}[1]{\mathinner{\langle #1|}}
\newcommand{\ket}[1]{\mathinner{|#1\rangle}}

\newcommand{\dyad}[1]{| #1\rangle \langle #1|}  % alias for \braket{}

\renewcommand{\a}{\alpha}
\renewcommand{\b}{\beta}

\def\QECC(#1,#2,#3,#4){(\!(#1,{#4}^{#2},#3)\!)_{#4}}
\def\stabQECC(#1,#2,#3,#4){[\![#1,#2,#3]\!]_{#4}}

\renewcommand{\r}{\varrho}
\newcommand{\overbar}[1]{\mkern 1.5mu\overline{\mkern-1.5mu#1\mkern-1.5mu}\mkern 1.5mu} % a wider & thicker overbar

\begin{document}

\title  {Quantum Codes of Maximal Distance and Highly Entangled Subspaces}
\date   {June 15, 2020}

\author {Felix Huber}
\affiliation{ICFO - The Institute of Photonic Sciences, 08860 Castelldefels (Barcelona), Spain}
\orcid{0000-0002-3856-4018}
\email      {felix.huber@icfo.eu}
\affiliation{Institut für Theoretische Physik, Universität zu Köln, 50937 Köln, Germany}
\affiliation{Naturwissenschaftlich-Technische Fakultät, Universität Siegen, 57068 Siegen, Germany}

\author	{Markus Grassl}
\orcid{0000-0002-3720-5195}
\affiliation{International Centre for Theory of Quantum Technologies, University of Gdansk, 80-308 Gda\'nsk, Poland}
\affiliation{Max Planck Institute for the Science of Light, 91058 Erlangen, Germany}
\email{markus.grassl@ug.edu.pl}

\begin{abstract}
We present new bounds on the existence of general quantum maximum
distance separable codes (QMDS): the length \(n\) of all QMDS codes 
with local dimension \(D\) and distance \(d \geq 3\) is bounded by
\(n \leq D^2 + d - 2\). We obtain their weight distribution and
present additional bounds that arise from Rains' shadow inequalities.
Our main result can be seen as a generalization of bounds that are
known for the two special cases of stabilizer QMDS codes and
absolutely maximally entangled states, and confirms the quantum MDS
conjecture in the special case of distance-three codes. As the
existence of QMDS codes is linked to that of highly entangled
subspaces (in which every vector has uniform \(r\)-body marginals) of
maximal dimension, our methods directly carry over to address
questions in multipartite entanglement.
 \end{abstract}
%% Keywords: quantum error correction, multipartite entanglement, quantum MDS codes.

\maketitle

\section{Introduction}
The processing of information with quantum particles is inevitably affected by disturbance from the environment.
By distributing the information onto many particles,
quantum error correcting codes (QECC) can safeguard quantum information from unwanted noise. 
In this way, a limited amount of corruption or even particle-loss can be tolerated.
Since the discovery of quantum error correction~\cite{PhysRevA.52.R2493, PhysRevLett.77.793}
and the establishment of its theoretical foundations~\cite{PhysRevA.55.900, PhysRevA.54.3824, PhysRevA.54.2629, 681315}, 
the search for ``good'' codes with desirable characteristics has been an ongoing endeavor. 
Both increasingly better-performing codes~\cite{doi:10.1063/1.1499754, PhysRevLett.97.180501, 6671468, Hastings2014, RevModPhys.87.307, Campbell2017} 
as well as stricter bounds imposed upon their existence have been found~\cite{761270, 7950914}.

The {\em quantum Singleton bound} can be seen as having its origins in the
no-cloning theorem~\cite{PhysRevA.56.1721, lecture_Gottesman}.
It states that the parameters of any quantum error correction code of distance~$d$,
encoding states from $\C^K$ into a subspace of $n$ systems with local dimensions $D$ each,
are bounded by
 \begin{equation}
  K \leq D^{n-2(d-1)}\,.
 \end{equation}
Codes achieving this bound are called quantum maximum distance
separable (QMDS) \cite{PhysRevA.55.900, 782103}.  Hugging the
fundamental limit of no-cloning, one can expect these codes to have
particularly intriguing features. 

The study of multipartite entanglement has led to the discovery of different types 
of entanglement that can be shared by three or more quantum particles~\cite{PhysRevA.62.062314, 1612.02437v2, Bengtsson_Zyczkowski_2017}. 
In turn, subspaces whose vectors show interesting entanglement properties have been investigated, 
such as those showing a bounded Schmidt rank~\cite{doi:10.1063/1.2862998}, 
having a negative partial transpose~\cite{PhysRevA.87.064302}, 
and being completely~\cite{1751-8121-41-37-375305, PhysRevA.90.062323} or genuinely entangled~\cite{
PhysRevA.98.012313, PhysRevA.100.062318, demianowicz2020approach}.

Generalizing the concept of maximal bipartite entanglement,
{\(r\)-uniform} states are a particular type of highly entangled pure
quantum states: these states exhibit maximal entanglement between any
\(r\) particles and the rest, in the sense that all of their \(r\)-sized
marginals are maximally mixed (i.e. uniform).  

It is reasonable to think that there are not too many states with this
property, and one might be tempted to ask the following question:
given a number \(n\) of \(D\)-level quantum systems, what is the
largest possible subspace in which every state vector is
\(r\)-uniform?  In other words, what is the dimension of the largest
possible {\em \(r\)-uniform subspace} (rUS), and by what methods can
this subspace be characterized?

It can be established that the concepts of so-called {\em pure} QECC
and rUS are in fact equivalent~\cite{PhysRevA.76.042309}.
Consequently, the attainable dimensions of both objects are
constrained by the quantum Singleton bound.  In this article, we will
focus on the case of QECC and rUS achieving this bound, that is, on
general QMDS codes and their corresponding highly entangled subspaces.  All of
our results can thus be seen as results concerning both coding and
entanglement theory, and we will use methods from quantum error
correction to answer questions in multipartite entanglement, and vice
versa.

While the quantum Singleton bound was one of the earliest bounds
obtained on quantum codes, not much more about the structural
properties of QMDS codes is known than what was already obtained by
Rains in Ref.~\cite{782103}.  Explicit constructions for stabilizer
QMDS codes from classical maximal distance separable codes
followed~\cite{quant-ph/0312164v1, 7282626, 5961827}, and QMDS codes
were later understood to constitute optimal ramp secret sharing
schemes~\cite{Helwig_Thesis}.

It turns out that there are stronger constraints for the existence of
QMDS codes than their parameters meeting the quantum Singleton bound,
the full set of which are not yet known.  Similarly, classical MDS
codes have been studied for more than half a century, but despite of
that, the exact conditions for their existence have not yet been
entirely resolved~\cite{Joyner:2011:SUP:2073646, Ball2012}.

In this article, we obtain two new bounds: first, we prove that for
{\em any} (both stabilizer and non-stabilizer) QMDS code [respectively
\((d-1)\)-uniform subspace satisfying the Singleton bound] to exist, 
the following condition has to be met,
 \begin{equation}
  n 	 \leq D^2  + d - 2\,.
 \end{equation}
The result is obtained by a systematic investigation of certain families of QMDS codes.
Second, we use Rains' shadow inequality 
to further restrict the allowed parameters in the case of small ``alphabets''. 
This can be seen as additional constraints that originate in the 
monogamy of entanglement~\cite{PhysRevA.98.052317}.

Furthermore, we derive the weight distribution of QMDS codes, a useful
tool for the analysis and characterization of codes; it is seen that
the weights are solely determined by the parameters of the code,
regardless how the code was constructed. Note that all quantum MDS
codes found to date are constructed from classical MDS codes, in
particular using the stabilizer theory.  Yet, even if one would find a
different construction for quantum MDS codes, their quantum weight
distribution has to match that of their classical MDS counterpart.
Hence it is an intriguing question whether or not there exist QMDS
codes that do not arise from any classical construction.

The structure of this article is as follows: connections between
quantum error correcting codes and highly entangled subspaces are
drawn in Sections~\ref{sect:QECC} and~\ref{sect:highly_ent_subspaces}.
Then, Sections~\ref{sect:weight_enum} and~\ref{sect:new_from_old}
introduce methods that are needed for the proofs that follow: the
machinery of quantum weight enumerators and descendance rules for pure
codes are presented.  Sections~\ref{sect:QMDS}
and~\ref{sect:QMDS_families} introduce quantum maximum distance
separable codes and the families formed thereof.  The weights of QMDS
codes are derived in Section~\ref{sect:QMDS_weights}.  This results in
bounds on the existence of QMDS codes
(Sections~\ref{sect:max_len_QMDS} and~\ref{sect:shadow_bounds}).  The
QMDS conjecture is treated in Section~\ref{sect:conjecture}, before
concluding in Section~\ref{sect:conclusions}.  The appendices contain
proofs of the quantum Singleton bound and an overview on previous
bounds for stabilizer QMDS codes and AME states.  This is followed by
detailed tables on known QMDS constructions and bounds on their
existence for small local dimensions.

\section{Quantum error correcting codes}\label{sect:QECC}
A quantum error correcting code \(\QQ= (\!(n,K,d)\!)_D\) is a
\(K\)-dimensional subspace of \(\!(\C^D)^{\ot n}\) such that every
error affecting at most \(d-1\) subsystems can either be detected or
acts trivially on the code, i.e., introduces at most a global phase
factor.  Here, the parameter $d$ is the distance and
a code with \(d \geq 2t+1\) allows to correct any error that 
affects up to \(t\) subsystems, e.g. the complete depolarization of any~$t$ subsystems.

Let us gently introduce some notation to make this precise: 
denote by \(\{e_a\colon a=0,\ldots,D^2-1\}\) an orthogonal operator basis for \(\C^D\) 
that includes the identity \(e_0 = \one\), such that \(\tr(e_a^\dag e_b) = D \delta_{ab}\).
By taking \(n\)-fold tensor products of elements in \(\{e_a\}\)
we obtain a so-called local error basis $\{E_{\bm{a}}\}$ on~\(\!(\C^D)^{\ot n}\)  
satisfying \(\tr(E_{\bm{a}}^\dag E_{\bm{b}}) = D^n \delta_{\bm{a}\bm{b}}\). 
The support of an error-operator~$E_{\bm{a}}$, that is, the subsystems it acts 
non-trivially on, is denoted by $\supp(E_{\bm{a}})$.
The weight of $E_{\bm{a}}$ is the size of its support, \(\wt(E_{\bm{a}}) = |\supp(E_{\bm{a}})|\).
Finally, let \(\{\ket{i_\QQ}\colon i=1,\ldots,K\}\) be a set of orthogonal unit vectors spanning \(\QQ\).
Then \(\Pi_\QQ = \sum \dyad{i_\QQ}\) is the projector onto the code space.

For \(\QQ\) to be a QECC with minimum distance~$d$, 
a necessary and sufficient criterion is for
\begin{equation}\label{def:code}
 \bra{i_\QQ} E_{\bm{a}} \ket{j_\QQ} = C(E_{\bm{a}})\, \delta_{ij}
\end{equation}
to hold for all pairs \(\ket{i_\QQ}, \ket{j_\QQ}\) and errors \(E_{\bm{a}}\)
of weight strictly less than~\(d\)~\cite{PhysRevLett.84.2525}.  Note
that \(C(E_{\bm{a}})\in\C\) is a constant that depends only on the specific error
\(E_{\bm{a}}\), but not on the vectors \(\ket{i_\QQ}\) and
\(\ket{j_\QQ}\). A code is called {\em pure} if \(C(E_{\bm{a}}) = \tr(E_{\bm{a}})
/ D^n=0\) for all \(E_{\bm{a}}\) with \(0<\wt(E_{\bm{a}}) < d\).  In other words, the
constant \(C(E_{\bm{a}})\) of pure codes vanishes for all non-trivial errors that
the code is designed to detect.

For one of the proofs that follow, we will also need an entropic
condition on quantum error correction: consider the purification of
$\r=\Pi_\QQ/K$ with a reference system~$R$ of
dimension~$K$, $\ket{\phi_\QQ} =\frac{1}{\sqrt{K}}
\sum_{i=1}^K\ket{i_R} \ot \ket{i_\QQ}$.  The von Neumann entropy of a
subsystem $I$ is given by $S(\varrho_I) = - \sum_i \lambda_i
\log(\lambda_i)$, where~$\lambda_i$ are the eigenvalues of the reduced
density matrix $\varrho_I$ for the subsystem~$I$.  For the code to have
distance~$d$, a necessary and sufficient condition is that $S_{RA} = S_R + S_A$
holds for every subsystem $A$ with $|A|<d$, that is, the reference system
$R$ and the subsystem $A$ are uncorrelated.  From the conditions on
equality in the strong subadditivity, an equivalent formulation is
that $\varrho_{RA} = \varrho_R \ot \varrho_A$ must hold.

%%%% Highly Entangled Subspaces
\section{Highly entangled subspaces}\label{sect:highly_ent_subspaces}
A pure state \(\ket{\phi}\), whose reductions onto \(r\) parties are all maximally mixed, 
is termed {\em \(r\)-uniform}. 
That is, \( \tr_{S^c}(\dyad{\phi}) \propto \one\) for every subset \(S \subseteq \{1,\dots,n\}\) 
of size \(|S| \leq r\), where \(S^c\) denotes its complement.
An {\em {\(r\)-uniform} subspace} (rUS) is a subspace of \(\!(\C^D)^{\ot n}\) 
in which every vector is at least \(r\)-uniform. In other words, every vector \(\ket{\phi}\) 
lying in an rUS satisfies that for all error operators with \(\supp(E_{\bm{a}}) \subseteq S\) 
where \(|S|\leq r\)
\begin{equation}
\bra{\phi} E_{\bm{a}} \ket{\phi} = \tr[ \underbrace{\tr_{S^c}(\dyad{\phi})}_{\propto \one} \tr_{S^c}(E_{\bm{a}})] = \tr(E_{\bm{a}}) / D^n \,.
\end{equation} 
Accordingly, \(\ket{\phi}\) is maximally entangled across any bipartition of size \(r\) vs. \(n-r\),
having the largest possible von Neumann entropy on the smaller reduction.

From the definition of a QECC in Eq.~\eqref{def:code}, it is not hard
to see that a pure code with parameters \( (\!(n,K,r+1)\!)_D\) implies
the existence of an \(r\)-uniform subspace of \( ( \C^D)^{\otimes n}\)
with dimension~\(K\).  In fact, the converse statement is also true:
the existence of an \(r\)-uniform subspace implies that of a pure QECC
of distance \(r+1\).  The proof is based on an equivalent condition
for a subspace \(\QQ\) to be a QECC, namely that the expectation value
 \begin{equation}\label{def:code2}
  \bra{\phi} E \ket{\phi} = C(E)
 \end{equation} 
is constant for all \(\ket{\phi}\) ranging over the subspace \(\QQ\)
and operators \(E\) with support on less than \(d\)~parties.  The
claim then follows by considering pure codes for which \(C(E) = 0\)
when $0<\wt(E) < d$.

The equivalence of Eq.~\eqref{def:code} and Eq.~\eqref{def:code2} has
already been established, and we sketch the
proof~\cite{quant-ph/0004072v1, PhysRevLett.83.648}: expanding
\(\ket{\phi}\) in the logical basis \(\{ \ket{i_\QQ} \}\) and~\(E\) in
a Hermitean error basis \(\{E_{\bm{a}}\}\), Eq.~\eqref{def:code} implies
Eq.~\eqref{def:code2}.  The converse can be established by defining
the inner product \( \langle v,w \rangle_{E_{\bm{a}}} := \bra{v}
(E_{\bm{a}}+\lambda\one) \ket{w}\), and its associated norm \(||v||_{E_{\bm{a}}} =
\sqrt{\langle v, v \rangle_{E_{\bm{a}}}} \), where \(\lambda\geq 0\) is
chosen such that \(E_{\bm{a}}+\lambda\one \geq 0\).  With the complex
polarization identity
\begin{align}
\langle v, w \rangle_{E_{\bm{a}}} = \frac{1}{4} \big( \phantom{i}||v+\phantom{i}w||_{E_{\bm{a}}}^2 &- \phantom{i}||v-\phantom{i}w||_{E_{\bm{a}}}^2 \nn\\ 
 + \, i ||v-iw||_{E_{\bm{a}}}^2 & - i ||v+iw||_{E_{\bm{a}}}^2 \big)
\end{align} 
and the decomposition of \(\ket{i_\QQ}\) and \(\ket{j_\QQ}\) into sum and differences 
(with and without a complex phase \(i\)) 
of two vectors \(\ket{\psi}, \ket{\phi} \in \QQ\), 
it is seen that Eq.~\eqref{def:code2} implies Eq.~\eqref{def:code}.
Therefore, the formulations of Eq.~\eqref{def:code} and Eq.~\eqref{def:code2} are equivalent.

Considering these two definitions for the case of pure codes, one arrives at the following observation.
\begin{observation}[Equivalence of pure QECC and rUS]\label{prop:eq_QMDS_rUS}
 The following objects are equivalent:\\
 \noindent 1.  a pure \((\!(n,K,d)\!)_D\) quantum error correcting code;\\
 \noindent 2.  a \((d-1)\)-uniform subspace in \((\C^D)^{\ot n}\) of dimension~\(K\).
\end{observation}
 
Thus the question about the maximal dimension that an $r$-uniform
subspace can attain is one-to-one related to the maximal dimension of
pure codes.  In what follows we will mostly focus on pure codes, as
the corresponding results for uniform subspaces can simply be read off
Observation~\ref{prop:eq_QMDS_rUS}.

%%%% Weight Enumerator
\section{Weight enumerators}\label{sect:weight_enum}
We will make use of {\em weight enumerators}
in the proofs that follow. Their knowledge is not required to
understand the main result [Theorem~\ref{thm:QMDS_bound}] of this
article (if Theorem \ref{thm:reductions_of_codes} is accepted); in
that case this section can be skipped.

For classical codes, the weight enumerator counts the number of
codewords of a given Hamming weight.  Although there is no such direct
combinatorial interpretation of the quantum weight enumerator, it has
been shown \cite{PhysRevLett.78.1600, 681316} that quantum weight
enumerators are a useful tool for the characterization of quantum
codes and that they can, for example, be employed to determine their
distance, as well as to derive other properties of putative codes or
to show their non-existence.

Given a local error basis \(\{E_{\bm{a}}\}\) with \(\tr(E_{\bm{a}}^\dag E_{\bm{b}}) = D^n \delta_{\bm{a}\bm{b}}\), 
define the Shor-Laflamme weights of a code \(\QQ\) with associated projector \(\Pi_\QQ\) as~\cite{PhysRevLett.78.1600, 681316}
\begin{align}
 A_j(\Pi_\QQ) &= \sum_{\wt(E_{\bm{a}}) = j} \tr[E_{\bm{a}}^\dag \Pi_\QQ] \tr[E_{\bm{a}} \Pi_\QQ]\,, \label{eq:KLF_weight1}\\
 B_j(\Pi_\QQ) &= \sum_{\wt(E_{\bm{a}}) = j} \tr[E_{\bm{a}}^\dag \Pi_\QQ  E_{\bm{a}} \Pi_\QQ]\,. \label{eq:KLF_weight2}
\end{align}
The sum above is taken over all errors \(E_{\bm{a}}\) of weight~\(j\) in the basis.
Note that \(A_j = A_j(\Pi_\QQ)\) is simply the Hilbert-Schmidt norm of all correlations 
in the code that act on exactly \(j\)~parties non-trivially. 
Both \(A_j\) and \(B_j\) are non-negative quantities that are invariant under the action 
of local unitaries \(\UU_1 \ot \cdots \ot \UU_n\)~\footnote{
This can be seen from the fact that the purities \(\tr[\tr_{S^c}(\Pi_\QQ)^2]\) of reductions 
can be expressed in terms of the weights \(A_j\). 
In turn, the dual weights \(B_j\) can be expressed as linear combinations of \(A_j\); 
this follows from the quantum MacWilliams identity~\cite{PhysRevLett.78.1600, 681316, 1751-8121-51-17-175301}},
and thus do not depend on the specific orthonormal error basis chosen.

We will also need Rains' {\em unitary} weights~\cite{681316}, defined as
\begin{align}
 A'_j(\Pi_\QQ) &= \sum_{|S| =j} \tr [\tr_{S^c} (\Pi_\QQ ) \tr_{S^c} (\Pi_\QQ)]\,, \label{eq:unitary_weights1}\\
 B'_j(\Pi_\QQ) &= \sum_{|S| =j} \tr [\tr_{S} (\Pi_\QQ ) 	\tr_{S} (\Pi_\QQ)] \label{eq:unitary_weights2}\,,
\end{align} 
where the sum is over all subsets \(S \subseteq \{1,\dots,n\}\) of size~\(j\).
For readers familiar with measures in quantum information, 
these quantities are proportional to the average purities of suitably normalized reductions 
of size \(j\) and \(n-j\), respectively.\footnote{
If \(\Pi_\QQ\) is normalized to a quantum state \(\r = \Pi_\QQ/K\), 
\(A_j'(\varrho)\) is the sum over the purities of all \(j\)-body reductions.}
From the definition, \(A'_j = B'_{n-j}\).

A fine-graining of both types of weights will prove useful for later proofs:
\begin{align}
 \AA_S(\Pi_\QQ) &= \sum_{\supp(E_{\bm{a}}) = S} \tr[E_{\bm{a}}^\dag \Pi_\QQ] \tr[E_{\bm{a}} \Pi_\QQ]\,, \label{eq:fine_grained_KLF1}\\
 \BB_S(\Pi_\QQ) &= \sum_{\supp(E_{\bm{a}}) = S} \tr[E_{\bm{a}}^\dag \Pi_\QQ  E_{\bm{a}} \Pi_\QQ]\,. \label{eq:fine_grained_KLF2}\\
\AA'_S(\Pi_\QQ) &=  \quad \tr [\tr_{S^c} (\Pi_\QQ) \tr_{S^c} (\Pi_\QQ)]\,,  \label{eq:fine_grained_unitary1}\\
\BB'_S(\Pi_\QQ) &=  \quad \tr [\tr_{S} (\Pi_\QQ) \tr_{S} (\Pi_\QQ)]\,.  	\label{eq:fine_grained_unitary2}
\end{align}
These are simply the non-symmetrized versions of Eqs.~\eqref{eq:KLF_weight1} to \eqref{eq:unitary_weights2}
for a fixed subset~\(S\).

The following facts about the weights of codes are known~\cite{681316}: 
necessary and sufficient conditions for a projector \(\Pi\) of rank \(K\) to be a QECC of distance~\(d\) are
\begin{equation}\label{eq:code_cond_enum}
 K B_j(\Pi) = A_j(\Pi) \quad \text{for} \quad 0\leq j<d\,. 
\end{equation} 
These conditions can be restated in terms of the unitary
enumerators. The quantities \(A'_j\) and \(B'_j\) are linear functions
of the quantities \(A_{i}\) and \(B_{i}\) with \(i \leq j\)
respectively,\footnote{This can be established by writing the partial
  trace as a channel.  See Refs.~\cite{681316, 1751-8121-51-17-175301}
  for more details.}
\begin{align}\label{eq:unitary_decomposition}
  A'_j(\Pi) &= \sum_{i\leq j} D^{-j} \binom{n-i}{n-j} A_i(\Pi) \,, \\
  B'_j(\Pi) &= \sum_{i\leq j} D^{-j} \binom{n-i}{n-j} B_i(\Pi) \,.
\end{align}
(Note that this resembles the notion of binomial moments in \cite{AshikhminBarg99}.)
Hence the relations of Eq.~\eqref{eq:code_cond_enum} are equivalent to 
\begin{equation}\label{eq:code_cond_unitary_enum}
 K B'_j(\Pi) = A'_j(\Pi) \quad \text{for} \quad 0\leq j<d\,.
\end{equation} 
Generally, one has that \(K B_j \geq A_j\) and \(K B'_j \geq A'_j\) for all \(j\), while
\(KB_0 = A_0 = K^2\). 

Analogous relations hold for subsets.  Let \(T\) be a subset of size
less than~\(d\).  From the conditions for the image of a projector
to be a code subspace [Eq.~\eqref{eq:code_cond_enum}] it follows that
\(K \BB_T = \AA_T\), while generally \(K \BB_S \geq \AA_S\)
holds.\footnote{ For a detailed derivation of this fact, see
  Appendix~B in Ref.~\cite{1751-8121-51-17-175301}}  Similarly, it
can be seen that \(K \BB'_T = \AA'_T\) holds, while \(K \BB'_S \geq
\AA'_S\) for arbitrary subsets~\(S\).  (In terms of purities of the
normalized projector \(\r = \Pi_\QQ/K\), this simply amounts to \(K
\tr[\tr_{T}(\varrho)^2] \geq \tr[\tr_{T^c}(\varrho)^2]\), with equality for all
subsets $|T| < d$.)

From the definition in Eq.~\eqref{def:code}, it follows that {\em pure} codes are those with \(A_j = 0\) 
[or correspondingly, \(A'_j = \binom{n}{j}K^2  D^{-j}\)] for all~\(0<j<d\). 
These are codes whose spanning vectors have maximally mixed \( (d-1) \)-body marginals, 
and correspond to \(r\)-uniform subspaces.

\section{New codes from old}\label{sect:new_from_old}
To develop our main results, we need a method with which new codes can
be constructed from old ones. This is done by taking partial traces of
\(\Pi_\QQ\).
\begin{theorem}[Rains~\cite{681316}]\label{thm:reductions_of_codes}
Let \((\!(n,K,d)\!)_D\) be a pure code with \(n, d \geq 2\). Then
there exists a pure code \( (\!(n-1,DK,d-1)\!)_D\).
\end{theorem}
\begin{proof}
Let the code space be spanned by an orthogonal set of vectors, \(\Pi_\QQ = \sum_{i=1}^K \dyad{i_\QQ}\). 
For simplicity,  we normalize the projector onto the code space to a density matrix, \(\r = \Pi_\QQ / K\), 
such that \(\tr(\varrho) = 1\). 
(This is motivated by the fact that \(\r\) stays normalized after application of the partial trace).
The code being pure, it follows from Eq.~\eqref{def:code} that all marginals 
of the spanning vectors \(\ket{i_\QQ}\) on less than~\(d\) parties must be maximally mixed. 
Accordingly, the above vectors can for any subset of parties \(S\subseteq\{1,\dots,n\}\) 
with \(|S| < d\) be Schmidt-decomposed as
\begin{equation}
 \ket{i_\QQ} = \frac{1}{\sqrt{D^{|S|}}} \sum_{\ell = 1}^{D^{|S|}} \ket{v_i^{(\ell)}}_S \ot \ket{w_i^{(\ell)}}_{S^c} \,.
\end{equation} 
We will now show that after performing a partial trace over parties of some subset \(V\) with \(|V|<d\), 
the operator \( \tr_{V}(\varrho) \) forms again 
(a projector onto) a pure code of distance \(d-|V|\) and dimension \(KD^{|V|}\). 
First, note that the rank of \( \tr_{V}(\varrho) \) can be at most~\(KD^{|V|}\),  
while the complementary operator \(\tr_{V^c}(\varrho)\) is proportional to the identity.
Because the reduction onto \(V\) is maximally mixed, \(\AA'_V(\varrho) = 1 / D^{|V|} \). 
From the condition in Eq.\eqref{eq:code_cond_unitary_enum}, 
the complementary reduction must have \(\BB'_V(\varrho) = \tr[ \tr_V(\varrho)^2] = 1/(KD^{|V|})\). 
As the operator \( \tr_V(\varrho)\) can have a rank of at most \(D^{|V|}K\), 
it must indeed be proportional to a projector onto a subspace of dimension \(D^{|V|}K\).

In similar manner, we can establish that the code \(\tr_V(\Pi_\QQ)\) has a distance of \(d-|V|\). 
For this we must check 
the condition in Eq.~\eqref{eq:code_cond_unitary_enum}, namely
\(K B'_j[\tr_S(\Pi_\QQ)] = A'_j[\tr_S(\Pi_\QQ)]\) for \(0\leq j<d-|V|\).
Respectively, it suffices to show that \(K \BB'_S(\r') = \AA'_S(\r')\) 
for all \(|S| < d-|V|\), with \(\r' = \tr_V(\varrho)\). 
Let~\(T\) be a subset of size smaller than \(d-|V|\) with \(T^c \cap V = \varnothing\). 
Then \(\tr[\tr_{T \backslash V}(\r')^2] = \tr[\tr_{T}(\varrho)^2] = 1/(KD^{|T|})\).
On the other hand, \(\tr[\tr_{T^c}(\r')^2] = 1/D^{|T|}\), where \(T^c\) is now 
the complement of \(T\) in \(\{1,\dots,n\} \backslash V\). We conclude that
\(K \BB'_T(\r') = \AA'_T(\r')=1/D^{|T|}\) for all \(|T| < d-|V|\), and 
\(\tr_V \Pi_\QQ\) indeed spans a pure code of distance \(d-|V|\). 
The claim follows by setting \(|V|\) to a single party.
\end{proof}
As established in the above proof, a pure \(\QECC(n,k,d,D)\) code
spawns a family of new pure codes having parameters
\(\QECC({n-s},{k+s},{d-s},D)\) for all integers \(s\) in \(0\leq s <
d\).  In the case of stabilizer codes, the same result can be obtained
more straightforwardly~\cite{quant-ph/9705052v1,
  alsina2019absolutely}.

\begin{remark}
Naturally, this implies that if the existence of a pure \(\QECC(n,k,d,D)\) can be ruled out, 
then any pure \(\QECC({n+s},{k-s},{d+s},D)\) for \(0\leq s \leq k\) cannot exist either.
\end{remark}

Does this method of creating new codes from old by partial trace also work for codes that are not pure?
It is tempting to think that any given impure code \((\!(n,K,d)\!)_D\) 
may yield a \((\!(n-1,K',d-1)\!)_D\) with \(K < K' \leq DK\).
However, this does not seem to be straightforward: 
consider Shor's code, which is an impure code with parameters $(\!(9,2,3)\!)_2$. 
A partial trace on the last qubit yields a projector of rank $4$,
yet it does not form an $(\!(8,4,2)\!)_2$ code (such that $K'=DK$), 
as an analysis of its weight distribution shows.\footnote{
The Shor code is spanned by the vectors $(\ket{000} + \ket{111})^{\ot 3}$ 
and $(\ket{000} - \ket{111})^{\ot 3}$. 
It has the weights
$A = [4,\allowbreak 0,\allowbreak 36,\allowbreak 0,\allowbreak 108,\allowbreak 0,\allowbreak 300,\allowbreak 0,\allowbreak 576,\allowbreak 0]$ and $B = [2,\allowbreak 0,\allowbreak 18,\allowbreak 78,\allowbreak 54,\allowbreak 414,\allowbreak 150,\allowbreak 666,\allowbreak 288,\allowbreak 378]$
giving it a minimum distance $d=3$.
After a partial trace over a single particle, one obtains
$A' = [ 16,\allowbreak 0,\allowbreak 112,\allowbreak 0,\allowbreak 240,\allowbreak 0,\allowbreak 400,\allowbreak 0,\allowbreak 256]$ and $B' = [4,\allowbreak 8,\allowbreak 80,\allowbreak 152,\allowbreak 520,\allowbreak 568,\allowbreak 1136,\allowbreak 808,\allowbreak 820]$,
a trivial code with distance $d=1$ and $K=4$.}

\section{Quantum MDS codes}\label{sect:QMDS}
Let us recall the bound from which the concept of QMDS codes originates, the quantum Singleton bound.
\begin{theorem}[Rains~\cite{782103}]\label{thm:quantum_singleton_bound}
 Let \(\QQ = [\![n,k,d]\!]_D\) be a QECC. Its parameters are bounded by
 \begin{equation}\label{eq:qmds_bound}
  k +2d \leq n+2 \,.
 \end{equation}
\end{theorem}
For a code \(\QQ = (\!(n,K,d)\!)_D\) with \(K\) not necessarily a
power of \(D\), the quantum singleton bound reads
 \begin{equation}\label{eq:qmds_bound_general}
  K \leq D^{n-2(d-1)}\,.
 \end{equation}
Two proofs of the quantum Singleton bound are presented in Appendix~\ref{app:quantum_singleton_bound}.

A code that achieves equality in Eqs.~\eqref{eq:qmds_bound} and~\eqref{eq:qmds_bound_general}, respectively,
[i.e., having parameters \((\!(n,D^{n-2d+2},d)\!)_D\)] is called a {\em quantum maximum distance separable} code (QMDS). 
The length \(n\) of QMDS codes is unbounded
for \(d\leq 2\); these codes are called {\em trivial}~\cite{Brun_Lidar_QECC}. 
From now on, we restrict ourselves to non-trivial QMDS codes, 
and can make use of \(n+2 = k+2d\) in all derivations that follow. 

It happens that all QMDS codes are pure~\cite{782103, Brun_Lidar_QECC}.
For this fact we will present a new information theoretic proof which was
kindly communicated to us by Andreas Winter~\cite{Winter_private-comm_QMDS-purity_2019}.
The following lemma on the von Neumann entropy $S(J) = S(\varrho_J) = - \sum_i \lambda_i \log(\lambda_i)$ of a 
subsystem $J \subseteq \{1,\dots,n\}$, where $\lambda_i$ are the eigenvalues of $\varrho_J$, is needed.

\begin{lemma}[Winter~\cite{Winter_private-comm_QMDS-purity_2019}]\label{lemma:entropy}
Let $n \geq m > \ell$. Then
\begin{equation}
 \frac{1}{\binom{n}{m}}\sum_{\substack{I \subseteq \{1,\dots,n\}\\ |I|=m}} S(I) \,\leq \,
 \frac{m}{\ell} \frac{1}{\binom{n}{\ell}}\sum_{\substack{J \subset \{1,\dots,n\}\\ |J|= \ell}} S(J)\,.
\end{equation} 
\end{lemma} 
The proof can be found in Appendix~\ref{app:entropy_lemma}.

\begin{theorem}[Rains~\cite{782103}]\label{prop:QMDS_pure}
Let \(\QQ\) be a QMDS code. Then \(\QQ\) is pure.
\end{theorem}
\begin{proof} (Winter~\cite{Winter_private-comm_QMDS-purity_2019})
Purify the projector~$\Pi_\QQ$ onto the code space with a reference system~$R$ of dimension~$D^k$.
For any bipartition $A|B$ of $\{1,\dots,n\}$ with sizes $|A| = d-1$ and $|B| = n-d+1$, respectively,
\begin{equation}
 S(B) = S(RA) = S(R) + S(A)
\end{equation} 
must hold for $\Pi_\QQ$ to be a code of distance~$d$
[cf. Section~\ref{sect:QECC}].
Naturally, also 
\begin{equation}
 S(\overbar{B}) = S(R) + S(\overbar{A})\,,
\end{equation}
where $S(\overbar{A})$ and $S(\overbar{B})$ denote the average entropy of subsystems 
in $\{1,\dots, n\}$ of sizes $d-1$ and $n-d+1$, respectively.
Making use of Lemma~\ref{lemma:entropy}, one has that 
\begin{equation}\label{eq:QMDS_purity_inequality}
S(R) = S(\overbar{B}) - S(\overbar{A})
     \leq \frac{n-2(d-1)}{d-1}S(\overbar{A})\,.
\end{equation} 
For a quantum MDS code, $S(R) = k = n-2(d-1)$.
Thus to satisfy Eq.~\eqref{eq:QMDS_purity_inequality}, 
$S(A) = (d-1)\log(D)$ for all $A$ of size $(d-1)$ must hold. This proofs the claim.
\end{proof}

It is interesting to note that Eq.~\eqref{eq:QMDS_purity_inequality} presents a trade-off,
where large values of $d/n$ and $k/n$ go hand in hand with a highly entangled code space.

Quantum maximum distance separable codes being pure, we can extend Observation~\ref{prop:eq_QMDS_rUS} 
to the case of subspaces that meet the quantum Singleton bound:
\begin{observation}[QMDS codes and maximal rUS]
The following objects are equivalent:\\
 1.  an \( \QECC(n,{n-2d+2},d,D)\) QMDS code;\\
 2.  a \((d-1)\)-uniform subspace in \((\C^D)^{\ot n}\) of dimension \(n-2d+2\).
\end{observation}
These objects---quantum MDS codes and $r$-uniform subspaces of maximal dimension---are now the main focus of our attention. 
All results in the following sections apply to both objects.

\section{QMDS families}\label{sect:QMDS_families}
By Theorem~\ref{thm:reductions_of_codes}, the existence of a QMDS code
with distance~\(d\) leads to a family of QMDS codes with distances
\(d' \leq d\) (see Fig.~\ref{fig:qmds_family}). As an example, the
existence of a code having the parameters \(\QECC(6,0,4,2)\) yields
the chain \(\QECC(6,0,4,2) \Rightarrow \QECC(5,1,3,2) \Rightarrow
\QECC(4,2,3,2) \Rightarrow \QECC(3,3,1,2)\), where we refer to
\(\QECC(6,0,4,2)\) as the parent code.  Such a {\em QMDS family} is
solely determined by the parameter \(n+k\) with $k=\log_D K$ ($n+k=6$
in the above example), and we are interested in the highest achievable
distance \(\tilde d = (\tilde n- \tilde k)/2+1\) within any given
family.

\begin{figure}[tbp]
\large
  \begin{tabular}{c@{\kern1cm}c}
    $n+k=6$, $D=2$ &  $n+k=12$, $D=3$\\[1ex]
    $\begin{array}[t]{rcl}
      \begin{picture}(0,0)
        \put(0,3){\vector(0,-1){50}}
      \end{picture}
       & \QECC(6,0,4,2) & \exists\\
       & \QECC(5,1,3,2) & \exists\\
       & \QECC(4,2,2,2) & \exists\\
       & \QECC(3,3,1,2) & \exists
    \end{array}$ &
    $\begin{array}[t]{rcl}
      \begin{picture}(0,0)
        \put(0,3){\vector(0,-1){100}}
      \end{picture}
      & \QECC(12,0,7,3) & \not\!\exists\\
      & \QECC(11,1,6,3) & \not\!\exists\\
      & \QECC(10,2,5,3) & \not\!\exists\\
      & \QECC( 9,3,4,3) & \not\!\exists\\
      & \QECC( 8,4,3,3) & \exists\\
      & \QECC( 7,5,2,3) & \exists\\
      & \QECC( 6,6,1,3) & \exists
    \end{array}$
  \end{tabular}
 \caption{Two examples of QMDS families. Left: a qubit QMDS family
   with $n+k=6$. All its members, up to the parent code
   \(\QECC(6,0,4,2)\), exist. Right: a qutrit QMDS family with
   $n+k=12$. Its members are known for \(d\leq
   3\)~\cite{quant-ph/0312164v1}, while it follows from the shadow
   inequalities that no corresponding codes can exist for \(d>3\).
   See also Tables~\ref{tab:constructions} and
   \ref{tab:bounds_D=3} in
   Appendix~\ref{app:constructions_tables}.}
 \label{fig:qmds_family}
\end{figure}

Note that the reversal of such a chain of codes might not always be
possible: for example, the existence of a code \(\QECC(8,4,3,3)\)
does not imply the existence of a code \(\QECC(9,3,4,3)\).  Indeed, a
construction for the former is known, whereas the existence of the
latter can be excluded (see Section~\ref{sect:shadow_bounds} and
Tables~\ref{tab:constructions} and \ref{tab:bounds_D=3} in
Appendix~\ref{app:constructions_tables}).  Nevertheless, any QMDS
code~\(\QQ\) has the characteristics of same-sized reductions of any
of its hypothetical parent codes~\(\tilde \QQ\): the reductions
of~\(\QQ\) are proportional to projectors, whose ranks match those of
hypothetical reductions of~\(\tilde \QQ\), while forming QECC
themselves.  This structure of nested projectors makes QMDS codes 
both attractive from the perspective of coding and entanglement theory, 
but also non-trivial to construct.

A certain part of the chain of the QMDS codes, consisting of the two top-most codes in any family, 
can always be reversed.
\begin{proposition}\label{prop:purification}
 The existence of the following two QMDS codes is equivalent:
\begin{equation}
 (\!(n,1,n/2+1)\!)_D \q\Longleftrightarrow\q (\!(n-1,D,n/2)\!)_D \,.
\end{equation}
(Note that for these to be QMDS codes, \(n\) must necessarily be even).
\end{proposition}

\begin{proof}
\noindent $\Rightarrow$: This direction was established in Theorem~\ref{thm:reductions_of_codes}.

\noindent $\Leftarrow$: 
  Let us purify \((\!(n-1,D,n/2)\!)_D\) with the associated projector
  \(\Pi_\QQ = \sum_{i=1}^D \dyad{i_\QQ}\) to a state on \(n\) parties, 
\begin{equation}
 \ket{\phi} = \frac{1}{\sqrt{D}} \sum_{i=1}^D \ket{i_\QQ} \ot \ket{i_R}\,,
\end{equation} 
where \(\{\ket{i_R}\}\) is a basis for the \(n\)th particle.  From
the conditions in Eq.~\eqref{eq:code_cond_unitary_enum} it follows
that for \(\dyad{\phi}\) to be a pure code of distance \(n/2+1\) it
suffices to check that \(B'_j(\dyad{\phi}) = A'_j(\dyad{\phi}) =
\binom{n}{j}D^{-j}\) (as \(K=1\)) for all \(j<n/2+1\).  By partially
tracing over any \((n/2-1)\) parties of \(\Pi_\QQ\), we see that this
is indeed the case.  With Theorem~\ref{thm:reductions_of_codes}, any
code \(\QQ\) with parameters \(\QECC(n-1,1,n/2,D)\) can be reduced to
a pure \(\QECC(n/2, n/2, 1,D)\), the latter corresponding to the
identity matrix on \(n/2\) particles.  Thus every reduction of
\(\Pi_\QQ = \sum_{i=1}^D \dyad{v_i}_\QQ\) onto~\(n/2\) particles is
maximally mixed.  Correspondingly, any reduction of \(\ket{\phi}\) of
size \(n/2\) that does not include the last particle is maximally
mixed.  From the Schmidt decomposition for pure states, it follows
that any \(n/2\)-sized reduction that includes the last particle must
then be maximally mixed, too.  Thus~\(\dyad{\phi}\) forms a pure code
of dimension \(1\) and distance \(n/2+1\).  This completes the proof.
\end{proof}

Hence, not only can an \((\!(n-1,D,n/2)\!)_D\) code be obtained by
partial trace from an \((\!(n,1,n/2+1)\!)_D\), but the latter can
always be constructed by purification from the former.  As a
consequence, all QMDS codes with \(k=0,1\) come in pairs, such as,
e.g., the codes \((\!(6,1,4)\!)_D\) and \((\!(5,D,3)\!)_D\), which
exist for all local dimensions~\(D\)~\cite{782103}.

While it was previously known that every pure code of dimension \(D\)
can with the addition of a single \(D\)-dimensional system be purified
to a rank-one quantum state~\cite{681316}, the increase in distance for
the case of QMDS codes [Proposition~\ref{prop:purification}] appears
to be new.

It is natural to ask under what conditions 
other steps in the hierarchy can be reversed.
While we leave this question open for now, note that for some cases
\begin{equation}
    (\!(n-2,D^2,n/2-1)\!)_D \quad \centernot\Rightarrow \quad  (\!(n-1,D,n/2)\!)_D  \,.
\end{equation}
For example, there exists a $\QECC(6,2,3,3)$ code, 
yet a $\QECC(7,1,4,3)$ can be excluded with the methods of Section~\ref{sect:shadow_bounds} 
(also see Table~\ref{tab:bounds_D=3} in Appendix~\ref{app:constructions_tables}).

Let us make a small detour to propagation rules for classical codes.
From a linear code $[n,k,d]_q$, one can obtain a code $[n-1,k,d-1]_q$ 
by an operation called {\em puncturing} (deleting one coordinate), 
and a code $[n-1,k-1,d]_q$ by {\em shortening} 
(taking an appropriate subcode after deleting one coordinate)~\cite{HuffmanPless2003}. 
Both operations yield MDS codes when starting with an MDS code. 
On the other hand, puncturing (i.e. projectively measuring a single subsystem) 
a quantum code $(\!(n,K,d)\!)_D$ yields 
a code $(\!(n-1,K,d-1)\!)_D$~\cite{681315}. 
But even when starting from a QMDS code, 
the resulting code is, in general, no longer QMDS.  
The analogue of shortening of quantum codes preserves the property of
being a QMDS code, but it is more involved and not always
possible. Rains~\cite{782103} has given a criterion when a stabilizer
code $[\![n-s,k-s,d]\!]_q$ can be derived from a stabilizer code
$[\![n,k,d]\!]_q$ by shortening.

\section{The weights of quantum MDS codes}\label{sect:QMDS_weights}
The weights of classical MDS codes are fixed by their parameters~\cite{MacWilliams1981}, 
and it is natural to ask if a similar result might also hold for their quantum analogue. 
This is indeed the case.

\begin{theorem}\label{thm:QMDS_weights}
 The {\em unitary} weights \(A'_j\) of a general QMDS code \(\QQ = \QECC(n,k,d,D)\) are given by
 \begin{equation}
  A'_j(\Pi_\QQ) =  \binom{n}{j} D^{2k-\min(2\a-j, j)} \,,
 \end{equation} 
 where \(\a = (n+k)/2\).
\end{theorem}
\begin{proof}
 From the repeated application of Theorem~\ref{thm:reductions_of_codes},
 all reductions of size smaller than or equal to \(\a\) are proportional to identity.
 On the other hand, all reductions of size \(j > \a\), 
 being pure codes with parameters \(\QECC(j, 2\a-j, j-\a+1,D)\), 
 are also proportional to projectors.
 These, however, have a non-full rank of \(2\a-j\), 
 namely the dimension of their code space. 
 Summing over all reductions of size \(j\) and taking into account 
 the normalization \(\tr(\Pi_\QQ) = D^k\) yields the claim.
\end{proof}

To obtain the Shor-Laflamme weights \(A_j\), we make use of the combinatorial version 
of the Möbius inversion formula (see p.~$267$ in Ref.~\cite{Stanley:2011:ECV:2124415}). 
Denote by \(2^{[n]}\) the set of subsets of \(\{1,\dots, n\}\). Given the functions 
$f\colon 2^{[n]}\rightarrow\R$ and $g\colon 2^{[n]} \rightarrow \R$ with
\( g(S) = \sum_{T\subseteq S} f(T)\,,\)
then 
\begin{equation}\label{prop:Mobius}
 f(S) = \sum_{T\subseteq S} (-1)^{|S - T|} g(T)\,.
\end{equation}
Using Möbius inversion, one can determine the weight distribution of QMDS codes.
\begin{theorem}
 The {\em Shor-Laflamme} weights \(A_j\) of a general QMDS code \(\QQ = \QECC(n,k,d,D)\) are given by
 \begin{equation}
  A_j(\Pi_\QQ) = \binom{n}{j} \sum_{i=0}^j  \binom{j}{i}  (-1)^{j-i} D^{2k+i-\min(2\a-i, i)}\,,
 \end{equation} 
 where \(\a = (n+k)/2\).
\end{theorem}
\begin{proof}
In Eq.~\eqref{eq:fine_grained_KLF1} and \eqref{eq:fine_grained_unitary1}, 
we defined the fine-grained weights
\begin{align}
  \AA_S(\Pi_\QQ)  &= \sum_{\supp(E_{\bm{a}}) = S } \tr[E_{\bm{a}}^\dag \Pi_\QQ] \tr[E_{\bm{a}} \Pi_\QQ]\,, \\
  \AA'_S(\Pi_\QQ) &=  \quad \tr [\tr_{S^c} (\Pi_\QQ ) \tr_{S^c} (\Pi_\QQ)]\,.
\end{align}
As shown in \cite{1751-8121-51-17-175301}, they are related via
\begin{equation}
 \AA_S'(\Pi_\QQ) = D^{-|S|} \sum_{T \subseteq S} \AA_T(\Pi_\QQ)\,.
\end{equation}
We can accordingly make use of the Möbius inversion [Eq.~\eqref{prop:Mobius}] to obtain~\cite{681316},
\begin{equation}
 \AA_S(\Pi_\QQ) = \sum_{T \subseteq S}  (-1)^{|S| - |T|}  D^{|T|} \AA'_T(\Pi_\QQ)\,.
\end{equation} 
With \(A_j(\Pi_\QQ) = \sum_{|S|=j} \AA_S(\Pi_\QQ)\), one obtains the Shor-Laflamme weights
for QMDS codes
\begin{align}
 &A_j(\Pi_\QQ) = \sum_{|S|=j} \sum_{T \subseteq S} (-1)^{|S| - |T|}  D^{|T|} \AA'_T(\Pi_\QQ) \nn\\
		     &= \binom{n}{j} \sum_{i=0}^j  \binom{j}{i}  (-1)^{j-i} D^{2k+i-\min(2\a-i, i)}\,.
\end{align}
This ends the proof.
 \end{proof}

\begin{remark}
The same result could be obtained by the polynomial
transform~\cite{681316, 1751-8121-51-17-175301}
\begin{equation}\label{eq:poly_transform_uni_to_SL}
 A(x,y) = A'(x-y,Dy)
\end{equation} 
where 
\begin{align}
  A(x,y) &= \sum_{j=0}^n A_j x^{n-j} y^j \,, &
  A'(x,y) &= \sum_{j=0}^n A'_j x^{n-j} y^j\,. \nn
  \end{align}
\end{remark}

Let us point out that the weights of an \(\QECC(n,k,d,D)\) QMDS code are proportional to 
those of any \(n\)-sized reduction of a {\em hypothetical} QMDS parent code 
\((\!( n+k, 1,  \frac{n+k}{2} + 1)\!)_D\)~\footnote{This argument could be refined: the weights of any 
\(\QECC(n,k,d,D)\) QMDS code are proportional to a reduction of any of its hypothetical QMDS parents 
\(\QECC(n+k-s, s, \frac{n+k}{2} + 1 -s,D)\) for all \(1\leq s \leq k\).} (cf.~the next section):
this hypothetical parent code is represented by a pure state \(\dyad{\phi}\) that
has all its \((n+k)/2\)-body marginals maximally mixed. 
Accordingly, its \(j\)-sized reductions~\(\varrho_{(j)}\) have 
purity \(\tr[ \varrho_{(j)}^2] = D^{-\min(j,n+k-j)}\). 
Indeed, let \(\Pi_{\QQ'} = D^k \tr_V (\dyad{\phi})\) with \(|V|=k\)  
be proportional to an \(n\)-sized reduction of \(\ket{\phi}\).
Summing over the purities of all its marginals of size~\(j\), 
we obtain the weights of Theorem~\ref{thm:QMDS_weights}. 
We conclude that the unitary weights 
of a QMDS code \(\QECC(n,k,d,D)\) are proportional to those of any 
\(n\)-sized reduction \(\tr_V \dyad{\phi}\) with ($|V|=k$) 
of a hypothetical \((\!(n+k,1, \frac{n+k}{2})\!)_D\) code.

This observation motivates the bound on QMDS codes that follows in
Section~\ref{sect:max_len_QMDS}.  It is of the same type as the one
obtained by Scott in Ref.~\cite{PhysRevA.69.052330} (see
Proposition~\ref{prop:AME_bound} in Appendix~\ref{app:stab_AME}) for
the existence of absolutely maximally entangled states.

\section{The maximal length of QMDS codes}
\label{sect:max_len_QMDS}
In this section we derive a new bound for the existence of QMDS codes.
Our bound generalizes a result by Ketkar et al.~\cite{1715533} 
on QMDS {\em stabilizer} or {\em additive} codes to QMDS codes of any type. 
It can equally be seen as a generalization of a bound by Scott~\cite{PhysRevA.69.052330} 
on the existence of codes with parameters \((\!(n,1,\nhf+1)\!)\) that are known as 
{\em absolutely maximally entangled} states or {\em perfect tensors}~\cite{PhysRevA.86.052335, PhysRevLett.118.200502}. 
Thus our main result extends Props.~\ref{prop:QMDS_stab_bound} and~\ref{prop:AME_bound} 
in Appendix~\ref{app:stab_AME} to {\em all} QMDS codes.

\begin{theorem}[Maximal length of QMDS codes]\label{thm:QMDS_bound}
 Let \(\QQ = \QECC(n,k,d,D)\) be a (stabilizer or non-stabilizer) QMDS code with \(d\geq 3\). Then
 \begin{align}
  n 	 &\leq D^2  + d - 2\,, \quad\quad \text{ or equivalently}  \\
  n+k &\leq 2(D^2-1)\,.
 \end{align}
\end{theorem}
\begin{proof}
Denote by \(\Pi_\QQ\) the projector onto the code space. For convenience 
we normalize the code to a quantum state \(\r = \Pi_\QQ /D^{k}\), such that  \(\tr[\r] = 1\).
Define \(\a = (n+k)/2\), and denote by~\(\r'\) the reduced density matrix of~\(\r\) 
corresponding to the code \(\QQ' = \QECC(a+1, {\a-1}, 2,D)\) 
on the first \(\a + 1\) parties. 
 Likewise, denote by~\(\r''\) a reduced density matrix of~\(\r\) corresponding 
 to \(\QQ'' = \QECC(\a+2, {\a-2}, 3,D)\) on the first \(\a+2\) parties.
 By Theorem~\ref{thm:reductions_of_codes} both~\(\QQ'\) and~\(\QQ''\) must be pure, 
 being derived from a pure code~\(\QQ\).
 Then \(\tr[\r'^2] = D^{-(\a-1)}\) and \(\tr[\r''^2] = D^{-(\a-2)}\). 
 Since \(A_{j} = 0\) for all \(0<j<\a+1\), we can decompose \(\r'\) and \(\r''\) in the Bloch representation in the following way,
 \begin{align}
  \r'  &= \frac{1}{D^{\a+1}} ( \one + P_{\a+1})\,, \\
  \r'' &= \frac{1}{D^{\a+2}} ( \one + \sum_{i=1}^{\a+2} P_{\a+1}^{(i)} \ot \one_i + P_{\a+2})\,,
 \end{align} 
 where \(P_{\a+1}\), \(P_{\a+1}^{(i)}\), and \(P_{\a+2}\) only contain terms of weight \(\a+1\), \(\a+1\), and \(\a+2\) respectively.
 Note that there are \(\a+2\) different terms \(P_{\a+1}^{(i)}\) with support on different subsystems in \(\r''\).
 Also, our normalization is chosen such that
 \begin{align}
  A_{\a+1}(\r') &= D^{-(\a+1)} \tr[P_{\a+1}^2] \,, \\
  A_{\a+1}(\r'') &= D^{-(\a+2)} \tr[(\sum_{i=1}^{\a+2} P_{\a+1}^{(i)}\ot \one_i)^2] \,, \\
  A_{\a+2}(\r'')& = D^{-(\a+2)}\tr[P_{\a+2}^2] \,, 
 \end{align}
 Making use of \(\tr[\r'^2] = D^{-(\a-1)}\)  leads to
\begin{equation}
 A_{\a+1}(\r') = D^2-1 \,.
\end{equation}
 Similarly, making use of \(\tr[\r''^2] = D^{-(\a-2)}\) yields,
 \begin{equation}
 A_{\a + 2}(\r'') = D^4 - (\a + 2)(D^2-1) -1 \geq 0 \,,
 \end{equation}
 which must be non-negative. Division by \((D^2-1)\) leads to
 \(D^2 - 1 - \a \geq 0\), which can be recast to the bound above.
 This proofs the claim. 
\end{proof}

\section{Shadow bounds}\label{sect:shadow_bounds}
Considering absolutely maximally entangled states, 
stronger bounds on their existence can be made than what is achieved by the bound 
from Scott [see Proposition~\ref{prop:AME_bound} in Appendix~\ref{app:stab_AME}] 
in case their local dimension is small~\cite{1751-8121-51-17-175301, HuberWyderka:ametable}.
In a similar spirit, it is possible to constrain the existence of low-dimensional QMDS codes further.

The {\em shadow inequalities} state that for any positive
semi-definite operators \(M_1,M_2\geq 0\) and any subset \(T\subseteq \{1,\dots,n\}\), 
the following family of inequalities hold~\cite{681316, 817508}.
\begin{equation}\label{eq:shadow_inequality}
 \sum_{S\subseteq\{1,\dots,n\}} (-1)^{|S\cap T|} \tr[ \tr_{S^c}(M_1) \tr_{S^c}(M_2)] \geq 0\,.
\end{equation}
The shadow inequalities can be seen as a family of monogamy of entanglement relations that constrain 
the entanglement appearing in the code subspace~\cite{PhysRevA.98.052317}.\footnote{
The shadow inequalities are also connected to a family of positive maps that generalize the reduction map \(\r \mapsto \one - \r\). 
Consequently, Eq.~\eqref{eq:shadow_inequality} also holds in operator form: for all \(M\geq~0\) and all 
subsets \(T\subseteq \{1,\dots,n\}\), 
\(\sum_{S\subseteq\{1,\dots,n\}} (-1)^{|S\cap T|} M_S \ot \one_{S^c} \geq 0\), 
where \(M_S = \tr_{S^c} M\)~\cite{PhysRevA.98.052317}.}
In order to use the shadow inequalities to determine the existence of codes, one sets \(M_1 = M_2 = \Pi_\QQ\) and checks the non-negativity 
of Eq.~\eqref{eq:shadow_inequality} for all subsets \(T\subseteq \{1,\dots,n\}\).

Let \(\QQ = \QECC(n,k,d,D)\) be a QMDS code. Then 
\(\tr_{S^c}(\Pi_\QQ)^2 = D^{2k+\min(n+k-|S|,|S|)}\), in line with the arguments of the proof of Theorem~\ref{thm:QMDS_weights}.
Thus their structure in terms of their unitary invariants is symmetric under permutation of the subsystems.
We thus do not forgo by considering a symmetrized version of Eq.~\eqref{eq:shadow_inequality} only,
the coefficients of the so-called shadow enumerator~\cite{681316,817508}
\begin{align} \label{eq:Shadow_coeff_in_Krawtchouk}
S_j(\Pi_\QQ) &= 
\sum_{|T^c|=j} \sum_{S\subseteq\{1,\dots,n\}} (-1)^{|S\cap T|} \tr[ \tr_{S^c}(\Pi_\QQ)^2] \nn\\
&= \sum_{\ell=0}^n K_{n-j}(\ell,n) A'_\ell(\Pi_\QQ) \geq 0\,.
\end{align}
Above, \(K_m(\ell,n)\) is the Krawtchouk polynomial defined as
\begin{equation}
 K_m(\ell,n) = \sum_{\b=0}^{m} (-1)^\b \binom{n-\ell}{m-\b} \binom{\ell}{\b}\,.
\end{equation}

\begin{remark}
The same result can be obtained by the polynomial transform~\cite{681316, 1751-8121-51-17-175301}
\begin{equation}\label{eq:poly_transform_uni_to_shadow}
 S(x,y) = A'(x+y,y-x)
\end{equation} 
where 
\begin{align}
  S(x,y) &= \sum_{j=0}^n S_j x^{n-j} y^j \,, &
  A'(x,y) &= \sum_{j=0}^n A'_j x^{n-j} y^j\,. \nn
  \end{align}
\end{remark}

We now can state the following corollary for the special case of QMDS codes.
\begin{corollary}\label{cor:shadow_bound}
Let \(\QECC(n,k,d,D)\) be a QMDS code. The following expression must be non-negative
for all $j$ in \(0\leq j \leq n\),
 \begin{equation}
 S_j = D^{2k} \sum_{\ell=0}^n K_{n-j}(\ell,n) \binom{n}{\ell} D^{-\min(n+k -\ell, \ell)} \geq 0\,. 
 \end{equation} 
\end{corollary}

Generally, the constraints imposed by Eqs.~\eqref{eq:shadow_inequality} 
and~\eqref{eq:Shadow_coeff_in_Krawtchouk} do not seem to give rise to 
simple closed-form expressions on the existence or minimum distance of codes. 
In the case of a binary alphabet however, the constraints yield the following bounds 
(cf. Theorem~$15$ in Ref.~\cite{796376} and Theorem~$13.4.1$ in Ref.~\cite{Nebe2006}):
the minimum distance of pure codes  $(\!(n, 1, d)\!)_2$ is bounded by
\begin{align}
 d &\leq 
 \begin{cases}
  2\lfloor \frac{n}{6} \rfloor + 3 \q\q\text{ if } n = 5 \bmod 6; \\
  2\lfloor \frac{n}{6} \rfloor + 2 \q\q\text{ otherwise}\,,
 \end{cases}
\intertext{whereas the minimum distance of codes $(\!(n,K,d)\!)_2$ with $K>1$ is bounded by}
 d &\leq 
 \begin{cases}
  2\lfloor \frac{n+1}{6} \rfloor + 2 \q\text{ if } n = 4 \bmod 6; \\
  2\lfloor \frac{n+1}{6} \rfloor + 1 \q\text{ otherwise} \,.
 \end{cases}
\end{align}

As done for the case of AME states in Ref.~\cite{1751-8121-51-17-175301},  
it is possible to evaluate Corollary~\ref{cor:shadow_bound} numerically 
for any QMDS code having small enough parameters.
This leads to new bounds on the existence of QMDS codes in dimensions \(D\leq 5\), see Appendix~\ref{app:constructions_tables}.

\section{QMDS conjecture}\label{sect:conjecture}

The following conjecture relating the maximal length of QMDS codes and
the local dimension $D$ is of interest.  It follows from the classical
MDS conjecture, and thus concerns itself with QMDS codes of stabilizer
type only.  The MDS conjecture for classical codes states that
the length of a non-trivial linear MDS code over the field $GF(q)$ is bounded
by $n\le q+1$, with the exception of $q=2^m$ where
additionally codes with parameters $[q+2,3,q]_q$ and $[q+2,q-1,4]_q$
exist.  \cite[Research Problem (11.4)]{MacWilliams1981}.  Applying
this conjecture to the classical codes over the field $GF(q^2)$
corresponding to stabilizer codes, one obtains the following
conjecture.
\begin{conjecture}[QMDS Conjecture, Cor.~$65$ in Ref.~\cite{1715533}]\label{conj:QMDS}
 With exception of \([\![D^2+2, D^2-4, 4]\!]_D\)
 with \(D=2^m\) where $n\leq D^2+2$~(cf. Thm. 14 in Ref.~\cite{7282626}),
 the length of all {\em stabilizer} QMDS codes with \(d \geq 3\) is bounded by 
 \( n\leq D^2+1\).\footnote{
 Note that while in Corollary~$65$ of Ref.~\cite{1715533} the case $d=q^2$ and $n \le q^2+2$ is 
listed as well, this is already excluded by the quantum Singleton bound if $q>2$.}
\end{conjecture}

The strongest confirmation of the classical MDS conjecture was proven
in a seminal work by Ball, which showed that the conjecture is true
for linear $q$-ary codes when $q$ is prime~\cite{Ball2012}.
Even when the classical MDS conjecture turns out to be true,
Conjecture~\ref{conj:QMDS} could be violated by non-stabilizer QMDS
codes.  

On the other hand, our bound (Theorem~\ref{thm:QMDS_bound})
constrains the length of QMDS codes for distance \(d=3\) to \(n\leq
D^2+1\), confirming the QMDS Conjecture for this choice of distance.
For \(d=4\), our bound can be met when \(D=2^m\) (Thm. 14 in
Ref.~\cite{7282626}).  In general, however, Conjecture~\ref{conj:QMDS}
is still unresolved for \(d>3\).

From the bound in Theorem~\ref{thm:QMDS_bound} it is seen that QMDS
codes with distance \(d\geq 3\) can only exist if \(n+k\leq6\) for
qubits, \(n+k\leq16\) for qutrits, \(n+k\leq30\) for ququarts, and
\(n+k\leq 48\) in the case of local dimension \(D=5\).  Thus for
qubits, no other non-trivial QMDS codes exist apart from those
with the parameters of the known stabilizer codes \([\![6,0,4]\!]_2\)
and \([\![5,1,3]\!]_2\).  In the case of qutrits, only seven QMDS
families exist; for five of these, the optimal parent code has already
been found (see Table~\ref{tab:bounds_D=3}).

\section{Conclusions}\label{sect:conclusions}
It is readily seen that quantum maximum distance separable (QMDS) codes must
correspond to subspaces in which every unit vector shows maximal entanglement 
across all bipartitions where the smaller partition has size $(d-1)$.
The question under what conditions such codes exist is thus not only relevant 
in coding theory, but also for the study of multipartite entanglement.

Interestingly, all QMDS codes can be grouped into QMDS families whose
members can be regarded as being obtained by partial trace from a
(possibly hypothetical) parent code of larger length and distance.
Since all descendants within a QMDS family form codes of smaller
distance themselves, their spectra are completely determined by the
parameters of their parent code.  This insight completely determines
the weight enumerator of QMDS codes.  It also leads to a bound
applicable to all (stabilizer and non-stabilizer) QMDS codes,
extending results known for the special cases of stabilizer QMDS codes
and absolutely maximally entangled states.  Moreover, the application
of Rains' shadow inequalities yields additional non-existence results.

The quantum Singleton bound is independent of the local dimension
\(D\) and one thus cannot expect it to be particularly strong.
However, if the Singleton bound can be met, classical codes are in all
known cases origin of these optimal quantum codes and highly entangled
subspaces.  More precisely, the majority of non-trivial QMDS codes in
the literature are of stabilizer type and hence based on classical
additive or linear MDS codes. There are also some examples of
non-stabilizer (also called {\em non-additive}) QMDS codes, in
particular codes of distance two \cite[Thm.~$7$]{746807}.  However,
putting these non-additive codes into the framework of
so-called union stabilizer codes or CWS codes (see
\cite[Chapter~10]{Brun_Lidar_QECC}), one finds a connection to
classical non-additive MDS codes as well.

It is an open question if this must generally be the case
and we state the problem more formally:
\begin{researchproblem}
 Is every quantum maximum distance separable code
 related to a classical MDS code?
\end{researchproblem}
It is indeed intriguing that hitherto no genuine ``quantum'' constructions
have been found that surpass their classical counterparts for these
types of codes. We note that an affirmative answer to this question 
would also reduce the question of the existence of absolutely maximally 
entangled states for some given even number of parties and local dimension 
to merely a finite computational problem (also see Problem $3$. in Ref.~\cite{KCIK2019}).

An interesting aspect seen here is that ``optimal'' codes
that have the largest possible distance 
must necessarily also exhibit the highest possible bipartite entanglement 
amongst the constituent particles.
One can readily expect a trade-off to be present, 
where large values of $k/n$ and $d/n$ necessarily go hand in 
hand with a highly entangled code space, 
whereas lowly entangled subspaces can only yield low values.
Indeed, such a trade-off can be seen in Eq.~\eqref{eq:QMDS_purity_inequality},
quantified by the average entropy of entanglement. 
A precise understanding of this trade-off might pave the way to 
derive stronger bounds on the performance of quantum codes, 
and could possibly help to explain the distance scalings found 
in low-density parity check codes~\cite{6671468}.

To conclude, QMDS codes present themselves as a rich playground: 
they form nested subspaces that are highly entangled and prove to be 
a testing ground for our understanding of multipartite entanglement. 
The discovery of further monogamy relations as well as entropic and 
rank inequalities would likely find an immediate application in stronger 
bounds on the existence of these ideal quantum objects.

\section{Acknowledgments}
We thank Daniel Alsina and Simeon Ball for fruitful discussions, 
Alexander Müller-Hermes for comments on the Entropy Lemma,
and Andreas Winter for kindly communicating his proof.
FH thanks David Gross and Otfried Gühne for their support, 
during which significant part of this work was carried out.
This work was supported by 
the Swiss National Science Foundation (Doc.Mobility 165024), 
the ERC (Consolidator Grant 683107/TempoQ),
the DFG (SPP1798 CoSIP),
the Excellence Initiative of the German Federal and State Governments (Grant ZUK 81), 
the Spanish MINECO (QIBEQI FIS2016-80773-P and Severo Ochoa SEV-2015-0522),
Generalitat de Catalunya (SGR 1381 and CERCA Programme),
and the Fundació Privada Cellex.
MG acknowledges partial support by the Foundation for Polish Science
(IRAP project, ICTQT, contract no. 2018/MAB/5, co-financed by EU
within the Smart Growth Operational Programme).

% balancing of columns
\onecolumngrid
\bigskip
\bigskip
\noindent\hrulefill
\bigskip

\twocolumngrid
\appendix

\section{Proofs of the quantum Singleton bound}
\label{app:quantum_singleton_bound}
We present two known proofs for the quantum Singleton bound below.
\begin{theorem}[Quantum Singleton bound~\cite{782103, PhysRevA.56.1721,Ashikhmin1997,AshikhminLitsyn99}]
 Let \((\!(n,K,d)\!)_D\) be a QECC. Its parameters are bounded by
 \begin{equation}
  K \leq D^{n-2(d-1)}\,.
 \end{equation}
\end{theorem}

\begin{proof}[Proof 1: (Rains, Thm.~$2$ in Ref.~\cite{782103})]
 Let us first show that \(2(d-1) \leq n\). Assume that $2(d-1)>n$ and consider $K=1$: 
 by convention, codes with \(K=1\) are only considered codes if they are pure, 
 and thus must have \(\tr_{S^c}\dyad{\phi} \propto \one\) for all \(|S|<d\). 
 From the Schmidt decomposition however it is seen that it is impossible that 
 marginals of size \(\nhf+1\) are of full rank, and thus \(2(d-1) \leq n\).
 Consider now $K>1$: in terms of the unitary weight enumerators, the conditions 
 for a projector \(\Pi_\QQ\) to be a QECC subspace read \(KB'_j(\Pi_\QQ) = A'_j(\Pi_\QQ)\) 
 for all \(j<d\) [Eq.~\eqref{eq:code_cond_unitary_enum}]. Also recall that by definition
 \(A'_j = B'_{n-j}\). If $2(d-1)>n$, one thus requires that
 \begin{equation}
  A'_{d-1} = K B'_{d-1} = K A'_{n-(d-1)}\,,
 \end{equation}
 and, due to \(n-(d-1) < (d-1)\), also that
 \begin{equation}
  A'_{n-(d-1)} = KB'_{n-(d-1)} = KA'_{d-1}\,,
 \end{equation}
 leading to a contradiction also for \(K>1\).

 Consequently, \(2(d-1) \leq n\). With the decompositions from Eq.~\eqref{eq:unitary_decomposition}, 
 one has that
 \begin{equation}
  A'_{d-1} = D^{-d+1}\sum_{i=0}^{d-1} \binom{n-i}{n-d+1} A_i\,,
 \end{equation}
 but also\footnote{The arXiv version (\href{https://arxiv.org/abs/quant-ph/9703048}{quant-ph/9703048}) of Ref.~\cite{782103} 
 contains in the corresponding formulae erroneous factors of \((D-1)^i\).}
 \begin{align}
  A'_{d-1} &= KB'_{d-1} = KA'_{n-d+1} \nn\\
  &= K D^{-n+d-1} \sum_{i=0}^{n-d+1} \binom{n-i}{d-1} A_i\,.
 \end{align}
With \(\binom{n-i}{n-d+1} = \binom{n-i}{d-1-i}\), the quantum Singleton bound follows 
from the analysis of 
 \begin{align}\label{eq:Singleton_sum}
  0 = A'_{d-1} - A'_{d-1} &=  K D^{-n + d-1} \sum_{i=0}^{n-d+1} \binom{n-i}{d-1}  A_i \nn\\
  &- D^{-d+1} \sum_{i=0}^{d-1} \binom{n-i}{d-1-i} A_{i}\,.
 \end{align} 
 Consider the coefficient \(A_i\) for \(0\leq i<d \). 
 \begin{equation} \label{eq:Singleton_coeff}
  K D^{-n+d-1} \binom{n-i}{d-1} - D^{-d+1} \binom{n-i}{d-1-i}\,.
 \end{equation} 
Note that 
\begin{align}
 &  \frac{\binom{n-i}{d-1}}{\binom{n-i}{d-1-i}} 
 =  \frac{(d-1-i)! (n-d+1)!}{(d-1)! (n-i-d+1)!} \displaybreak\nn\\ 
 =\,& \frac{(n-d+1)(n-d)\cdots (n-d+2-i)}{(d-1)(d-2)\cdots(d-i)} \geq 1\,,
 \end{align}
 because $n-d+1 \geq d-1 $ as established previously.
 If \(K > D^{n-2(d-1)} > 1\), the expression~\eqref{eq:Singleton_coeff} must 
 be non-negative due to \(A_i\geq 0\),  and it is furthermore strictly positive 
 in the case of \(i=0\) due to \(A_0 = K\).
 Consequently Eq.~\eqref{eq:Singleton_coeff} can only vanish if at least \(K\leq D^{n-2(d-1)}\).
This proofs the claim.
 \end{proof}

\begin{proof}[Proof 2: (Cerf \& Cleve~\cite{PhysRevA.56.1721})]
For this proof we only consider the case $K>1$. 
Then the distance must be bounded by \(2(d-1) \leq n\), for if not, 
two copies of the encoded state could be recovered each from reductions
of size \(n-(d-1) < d-1\), violating the no-cloning theorem.

Let \(\Pi_\QQ = \sum_{i=1}^K \dyad{i_{\QQ}}\) be the projector onto the code space.
The purification with a reference system~\(R\) leads to, 
 \begin{equation}
  \ket{\psi_{QR}} = \frac{1}{\sqrt{K}}\sum_{i=1}^K \ket{i_{\QQ}}\ot \ket{i_R}\,,
 \end{equation} 
 where \(\ket{i_R}\) is any orthonormal basis for \(R\).  Recall that
 the von Neumann entropy is defined as \(S(\varrho) = -\tr \r \log \r\).
 Let us partition the code into the three subsystems \(A,B,C\), such
 that \(|A|=|B|=d-1\) and $|C|=n-2(d-1)$.  Then \(S_R = S(
 \tr_{ABC}[\r]) = \log(K)\).  As the code has distance~\(d\), any
 subsystem of size strictly smaller than $d$ cannot reveal anything
 about the reference system $R$: indeed the condition of \(\varrho_{RA} =
 \varrho_R \ot \varrho_A\) is known to be a necessary and sufficient condition
 for the subsystem $A$ to be correctable~\cite{NielsenChuang2011};
 this is also equivalent to
 \(S_{RA} = S_R + S_A\).  With the subadditivity of the von Neumann
 entropy, namely \(S_{12} \leq S_1 + S_2\), this leads to
 \begin{align}
  S_R + S_A &= S_{RA} = S_{BC} \leq S_B + S_C\,,  \\
  S_R + S_B &= S_{RB} = S_{AC} \leq S_A + S_C\,,
 \end{align}
where we used that the entropies of complementary subsystems are equal for a pure state.
 The combination of the above two inequalities yields
 \(\log(K) = S_R \leq S_C\leq \log \dim (\HH_C) = \log D^{n-2(d-1)}\). This proofs the claim.
 \end{proof}
 
A third proof of the quantum Singleton bound using linear
programming can be found in Refs.~\cite{Ashikhmin1997,AshikhminLitsyn99}.

\section{Entropy lemma}\label{app:entropy_lemma}

\begin{lemma}[Winter~\cite{Winter_private-comm_QMDS-purity_2019}]
Let $n \geq m > \ell$. Then
\begin{equation}
 \frac{1}{\binom{n}{m}}\sum_{\substack{I \subseteq \{1,\dots,n\}\\ |I|=m}} S(I) \,\leq \,
 \frac{m}{\ell} \frac{1}{\binom{n}{\ell}}\sum_{\substack{J \subset \{1,\dots,n\}\\ |J|= \ell}} S(J)\,.
\end{equation} 
\end{lemma} 
\begin{proof}

For any subset $T\subseteq\{1,\ldots,n\}$, 
denote by $X_T$ the combination of the subsystems 
$\{X_i \colon i\in T\}$.
We first aim to show that 
\begin{equation}\label{eq:lemma_eq_1}
 S(X_{\{1,\dots, n\}} ) \leq \frac{1}{n-1} \sum_{i=1}^n S(X_{\{1,\dots,n\} \backslash\{i\}})\,.
\end{equation} 
For this, we purify the state with a reference system~$R$.
Then Eq.~\eqref{eq:lemma_eq_1} is equivalent to
\begin{equation}
 (n-1)S(R) \leq \sum_{i=1}^n S(X_i R)\,.
\end{equation} 
Rewritten in terms of the conditional von Neumann entropy $S(A|B) = S(AB) - S(B)$ yields
\begin{equation}
 -S(R) \leq \sum_{i=1}^n S(X_i|R)\,.
\end{equation}
To see that this holds, note that
\begin{equation}
 -S(R) = S(X_{\{1,\dots ,n\}}|R) \leq \sum_{i=1}^n S(X_i|R)\,,
\end{equation}
where the equality follows from the fact that the state on the entire
system $X_1\dots X_n R$ is pure, and the inequality follows from
strong subadditivity.
The repeated application of Eq.~\eqref{eq:lemma_eq_1} yields
\begin{equation}
 \frac{1}{m} \frac{1}{\binom{n}{m}}\sum_{|I|=m} S(I) \leq \frac{1}{\ell} \frac{1}{\binom{n}{\ell}} \sum_{|J|=\ell} S(J)\,.
\end{equation}
This completes the proof.
\end{proof}

A slightly more general form of above Lemma, the ``quantum Shearer's inequality'',
and its history can be found in Ref.~\cite{doi:10.1063/1.4939560}.

\begin{table*}[btp]\footnotesize
\tabcolsep0.5\tabcolsep
\footnotesep0.5\footnotesep % decrease vertical spacing between footnotes
\begin{tabular}{@{}l@{\quad} l l@{\quad} l@{}}
\textbf{Some QMDS constructions} \vspace{0.1em} \\ 
\hline
short (Gilbert-Varshamov)~\cite[Cor.~$32$]{1715533}: && \([\![n,n-2d+2,d]\!]_q\) &\( 2\leq d \leq \lceil \frac{n}{2} \rceil, \q \binom{n}{d} \leq q^2-1 \)   \\
Euclidean (CSS)~\cite[Cor.~$1$]{1365393}:\footnote{In Theorem~$14$ of Ref.~\cite{quant-ph/0312164v1}, 
and subsequently also in the overview table of Ketkar et al.~\cite{1715533}, 
only the upper bound $n\le q$ is given.
The bound $n\le q+1$ follows from Corollary $1$ of Ref.~\cite{1365393} 
(note that there, the condition $q>2$ for $n=q+1$ is missing).}
&& \([\![n,n-2d+2,d]\!]_q\) 		& \(1\leq d \leq \nhf +1, \q 3 \leq n \leq q+1\) 
for $2<q$\\
punctured GRM~\cite[Thm.~$5$ \& Cor.~$6$]{1523494}: && \([\![q^2-q\a, q^2-q\a-2d+2, d]\!]_q\)  & \(2 \leq d \leq q, \q 0\leq \a \leq q-d+1\)\\
Hermitean~\cite[Thm.~$14$]{quant-ph/0312164v1}:\footnote{%
Further details on what values $s$ can take can be found in Refs.~\cite{7282626, quant-ph/0312164v1}.}
&& \([\![q^2-s, q^2-s-2d+2, d]\!]_q\) 	& \(2 \leq d \leq q, \q 
\,\,\,\,\q\text{for some $0 \leq s<q$, incl. $s=0,1$}\)\\
single-error~\cite[Cor.~$3.6$]{5550401}:\footnote{%
The case where $q$ is odd also appears in Theorem~$1.1$ of Ref.~\cite{PhysRevA.82.052316}.}
&& \([\![n ,n-4, 3]\!]_q\) 			& $4 \leq n \leq q^2+1$\q except $q=2$ with $n=4$\\
Grassl/Rötteler I~\cite[Thm.~$13$]{7282626}:
&&\( [\![q^2 + 1, q^2 -2d + 3, d]\!]_q\) 	& \(2\leq d\leq q+1\), \q for $q$ odd or ($q$ even and $d$ odd)\\
Ball~\cite[Thm.~$4$]{Ball2019}:
&&\( [\![q^2 + 1, q^2 -2d + 3, d]\!]_q\) 	& \(2\leq d\leq q+1\), \q for $d \neq q$ \\
Grassl/Rötteler II~\cite[Thm.~$14$]{7282626}: 	&& \([\![q^2+2 , q^2-4, 4]\!]_q\) 	& \(q= 2^m\)\\
trivial~\cite[Thm.~$12$]{7282626}: 	    	&& \([\![n,n-2,2]\!]_D\) 		&\(n\) even and $(D$ odd or a multiple of $4)$\\
\hline
\end{tabular}
\caption{\label{tab:constructions}
Some known QMDS constructions when \(q=p^\b\) is a power of prime (except for the trivial QMDS). 
The table is partially adopted from Ref.~\cite{1715533}. Apart from
the trival QMDS codes with $d=2$, all codes are stabilizer codes.}
\end{table*}

\section{QMDS stabilizer codes and AME states}\label{app:stab_AME}
In order to set the bound appearing in
Section~\ref{sect:max_len_QMDS}, Theorem~\ref{thm:QMDS_bound}
into context, we shortly state the previously known bounds on
stabilizer and largest-distance QMDS codes.

Stabilizer codes are constructed from Abelian subgroups of
{\em nice error bases} not containing a non-trivial multiple of the 
identity.\footnote{
A nice error basis is a set of unitary matrices $\{E_g\}$ corresponding 
to a group $G$
such that $E_1 = \one$ and $E_g E_h = \omega_{g,h}
E_{gh}$~\cite{1019848}.}
When the local dimension $D=q=p^m$ is a power of a prime, such
Abelian subgroups correspond to additive codes $\CC$ over the finite
field $\F_{q^2}$ with $q^2$ elements. Additionally, the code $\CC$ is
contained in its dual, $\CC\subseteq \CC^{\perp_a}$, with respect to
the trace-alternating form on $\F_{q^2}^n$, given by
\begin{alignat}{5}\label{eq:trace-alternating-form}
\langle x,y\rangle_a = \tr_{\F_q/\F_p}\left(\frac{x\cdot y^q-x^q\cdot 
y}{\beta^{2q}-\beta^q}\right).
\end{alignat}
Here $(\beta,\beta^q)$ is a normal basis of $\F_{q^2}$ over $\F_q$
and the trace function for $q=p^m$ is defined as
$\tr_{\F_q/\F_p}(a)=\sum_{i=0}^{m-1} a^{q^i}$.

Thus, in the context of quantum MDS codes of stabilizer
type, their correspondence to classical MDS codes is of
relevance.
\begin{proposition}[Thm. 15 and Lemma 61 in \cite{1715533}]
The existence of the following is equivalent:\\
\noindent 1. an $[\![n,n-2d+2,d]\!]_q$
QMDS stabilizer code; \\
\noindent 2. an $[n,d-1,n-d+2]_{q^2}$ additive MDS code
$\CC\subset\F_{q^2}^n$ that is contained in its dual, 
$\CC\subset\CC^{\perp_a}$.
\end{proposition}
Note that the dual code $\CC^{\perp_a}$ is also an MDS code with
parameters $[n,n-d+1,d]_{q^2}$.

The following is known on the maximal length of stabilizer QMDS codes.

\begin{proposition}[Maximal length of QMDS {\em stabilizer} codes, 
Thm. $63$ in Ref.~\cite{1715533}] \label{prop:QMDS_stab_bound}
 Let \(\QQ = [\![n,k,d]\!]_D\) be a QMDS stabilizer code with \(d\geq 3\)
 and where $D$ is a prime-power. 
 Then
 \begin{equation}
  n \leq D^2 + d - 2 \,.
 \end{equation} 
\end{proposition}

A pure state \(\ket{\phi_{n,D}}\) of \(n\) parties with local
dimension \(D\) each is called {\em absolutely maximally entangled}
(AME), if maximal entanglement is present across every bipartition.
Consequently, all its reductions to half of its parties are maximally
mixed.  AME states are pure codes with parameters
\((\!(n,1,\nhf+1)\!)_D\).  If \(n\) is even, these states are the
top-most member of a QMDS family, reaching the largest distance
allowed by the quantum Singleton bound. They are also known
as {\em perfect tensors}.  Scott obtained the following bound on the
existence of absolutely maximally entangled states.
\begin{proposition}[Maximal length of AME states, Eq.~$44$ in Ref.~\cite{PhysRevA.69.052330}]
\label{prop:AME_bound}
 Let \(\ket{\phi_{n,D}}\) be an absolutely maximally entangled state of \(n\geq 4\) parties of 
 local dimension \(D\) each. Then
 \begin{equation}
  n \leq \begin{cases}
		2(D^2-1) 		&\text{if $n$ is even}; \\
		2D(D+1) -1 	&\text{if $n$ is odd} \,.
	    \end{cases}
 \end{equation}
\end{proposition}
Thus for \(n\) even, Proposition~\ref{prop:AME_bound} is indeed a bound on the existence of QMDS codes that have \(k=0\).

\section{Known constructions and tables}\label{app:constructions_tables}

We list parameters of some known QMDS constructions in
Table~\ref{tab:constructions}.  Tables~\ref{tab:bounds_D=3} to
\ref{tab:bounds_D=5} report on the highest distances within a QMDS
family that are not excluded by our bounds, as well as on the
parameters that can be reached by known constructions.  
All upper bounds listed arise from the shadow inequalities
(Corollary~\ref{cor:shadow_bound}). For local dimensions \(D > 5\), 
these constraints do not seem to be stronger than
those of Theorem~\ref{thm:QMDS_bound} and thus our tables only
include codes up to~\(D=5\).

Should the upper and lower bound meet, the corresponding code is
\emph{optimal} and specifies its QMDS family completely; these entries
are marked by~\(\ast\)\,.  Since all currently known non-trivial
constructions are stabilizer codes, we use the notation
$[\![n,k,d]\!]_q$ for both the lower and upper bound.  The upper bound
is valid for general codes $(\!(n, D^k, d)\!)_q$ as well.
\vskip12mm

%trick to get desired placement of floats
\onecolumngrid
\
\rule{0pt}{15mm}

\
\twocolumngrid

%%% D = 3
\begin{table}[ht]\tabcolsep0.5\tabcolsep\small
\begin{tabular}{@{}c@{\quad}lllll@{}}
\hline
\(\mathbf{n+k}\) 	& \textbf{upper} & \textbf{lower} &&&\\
\hline
\( 4 \)     &\([\![ 4 , 0 , 3 ]\!]_ 3 \)    &\([\![ 4 , 0 , 3 ]\!]_3 \)  &$\ast$&& Hermitean \\
\( 6 \)     &\([\![ 6 , 0 , 4 ]\!]_ 3 \)    &\([\![ 6 , 0 , 4 ]\!]_3 \)  &$\ast$&& Rains~\cite{782103}\\
\( 8 \)     &\([\![ 6 , 2 , 3 ]\!]_ 3 \)    &\([\![ 6 , 2 , 3 ]\!]_3 \)  &$\ast$&& single-error \\
\( 10 \)    &\([\![ 10 , 0 , 6 ]\!]_ 3 \)   &\([\![ 10, 0 , 6 ]\!]_3 \)  &$\ast$&& Glynn code~\cite{GLYNN198643}\\
\( 12 \)    &\([\![ 8 , 4 , 3 ]\!]_ 3 \)    &\([\![ 8 , 4 , 3 ]\!]_3 \)  &$\ast$&& single-error \\
\( 14 \)    &\([\![ 11 , 3 , 5 ]\!]_ 3 \)   &\([\![ 10, 4 , 4 ]\!]_3 \)  &      && Grassl/Rötteler I\\
\( 16 \)    &\([\![ 11 , 5 , 4 ]\!]_ 3 \)   &\([\![ 10, 6 , 3 ]\!]_3 \)  &      && single-error \\ 
\hline
\end{tabular}
\caption{\label{tab:bounds_D=3} Upper and lower bounds for the highest distance in QMDS families of local dimension~\(D=3\).}
\end{table}

%%% D = 4
\begin{table}[htbp]\tabcolsep0.5\tabcolsep\small
\begin{tabular}{@{}c@{\quad}lllll@{}}
\hline
\(\mathbf{n+k}\) 	& \textbf{upper} & \textbf{lower}	&&\\
\hline
\( 4 \)     &\([\![ 4   , 0 , 3 ]\!]_4 \)   &\([\![ 4  , 0 , 3 ]\!]_4 \)    &$\ast$ & Hermitean \\
\( 6 \)     &\([\![ 6   , 0 , 4 ]\!]_4 \)   &\([\![ 6  , 0 , 4 ]\!]_4 \)    &$\ast$ & Rains~\cite{782103} \\
\( 8 \)     &\([\![ 8  , 0 , 5]\!]_4 \)     &\([\![ 6  , 2 , 3 ]\!]_4 \)    && single-error\\
\( 10 \)    &\([\![ 10 , 0 , 6 ]\!]_4 \)    &\([\![ 10 , 0 , 6 ]\!]_4 \)    &$\ast$ & Gulliver et al. \cite{4608969} \\
\( 12 \)    &\([\![ 10 , 2 , 5 ]\!]_4 \)    &\([\![ 9  , 3 , 4 ]\!]_4 \)    && Grassl/Rötteler~\cite{7282626}\\
\( 14 \)    &\([\![ 14 , 0 , 8 ]\!]_4 \)    &\([\![ 10 , 4 , 4 ]\!]_4 \)    && shortening $[\![18,12,4]\!]_4$\\
\( 16 \)    &\([\![ 13 , 3 , 6 ]\!]_4 \)    &\([\![ 11 , 5 , 4 ]\!]_4 \)    && Grassl/Rötteler~\cite{7282626}\\
\( 18 \)    &\([\![ 18 , 0 , 10]\!]_4 \)    &\([\![ 12 , 6 , 4 ]\!]_4 \)    && shortening $[\![18,12,4]\!]_4$\\
\( 20 \)    &\([\![ 16 , 4 , 7 ]\!]_4 \)    &\([\![ 12 , 8 , 3 ]\!]_4 \)    && single-error \\
\( 22 \)    &\([\![ 22 , 0 , 12]\!]_4 \)    &\([\![ 14 , 8 , 4 ]\!]_4 \)    && shortening $[\![18,12,4]\!]_4$ \\
\( 24 \)    &\([\![ 19 , 5 , 8 ]\!]_4 \)    &\([\![ 14 , 10, 3 ]\!]_4 \)    && single-error \\
\( 26 \)    &\([\![ 23 , 3 , 11 ]\!]_4 \)   &\([\![ 17 , 9 , 5 ]\!]_4 \)    && Grassl/Rötteler I\\
\( 28 \)    &\([\![ 22 , 6 , 9 ]\!]_4 \)    &\([\![ 16 , 12, 3 ]\!]_4 \)    && single-error \\
\( 30 \)    &\([\![ 26 , 4 , 12 ]\!]_4 \)   &\([\![ 18 , 12, 4 ]\!]_4 \)    && Grassl/Rötteler II\\
\hline
\end{tabular}
\caption{\label{tab:bounds_D=4} Upper and lower bounds for the highest distance in QMDS families of local dimension~\(D = 4\).}
\end{table}

%%% D = 5
\begin{table}[htbp]\tabcolsep0.5\tabcolsep\small
\begin{tabular}{@{}c@{\quad}lllll@{}}
\hline
\(\mathbf{n+k}\) 	& \textbf{upper} & \textbf{lower} & &\\
\hline
\( 4 \)     &\([\![ 4 , 0 , 3 ]\!]_5 \)   	 &\([\![ 4 , 0 , 3 ]\!]_5 \)  &$\ast$& Hermitean\\ 
\( 6 \)     &\([\![ 6 , 0 , 4 ]\!]_5 \)   	 &\([\![ 6 , 0 , 4 ]\!]_5 \)  &$\ast$& Rains~\cite{782103}\\
\( 8 \)     &\([\![ 8 , 0 , 5 ]\!]_5 \)   	 &\([\![ 8 , 0 , 5 ]\!]_5 \)  &$\ast$& Kim/Lee~\cite{KimLee04}\\
\( 10 \)    &\([\![ 10 , 0 , 6 ]\!]_5 \)   	 &\([\![ 10 , 0 , 6 ]\!]_5 \) &$\ast$& Kim/Lee~\cite{KimLee04}\\
\( 12 \)    &\([\![ 12 , 0 , 7 ]\!]_5 \)   	 &\([\![ 10 , 2 , 5 ]\!]_5 \) && shortening $[\![26,18,5]\!]_5$\\
\( 14 \)    &\([\![ 14 , 0 , 8 ]\!]_5 \)   	 &\([\![ 14 , 0 , 8 ]\!]_5 \) &$\ast$& Ball~\cite{Ball2019}\\
\( 16 \)    &\([\![ 16 , 0 , 9 ]\!]_5 \)   	 &\([\![ 12 , 4 , 5 ]\!]_5 \) && shortening $[\![26,18,5]\!]_5$\\
\( 18 \)    &\([\![ 18 , 0 , 10 ]\!]_5 \)   	 &\([\![ 18 , 0 , 10]\!]_5 \) &$\ast$& Ball~\cite{Ball2019}\\
\( 20 \)    &\([\![ 20 , 0 , 11 ]\!]_5 \)   	 &\([\![ 14 , 6 , 5 ]\!]_5 \) && shortening $[\![26,18,5]\!]_5$\\
\( 22 \)    &\([\![ 22 , 0 , 12 ]\!]_5 \)   	 &\([\![ 15 , 7 , 5 ]\!]_5 \) && shortening $[\![26,18,5]\!]_5$\\
\( 24 \)    &\([\![ 24 , 0 , 13 ]\!]_5 \)   	 &\([\![ 16 , 8 , 5 ]\!]_5 \) && shortening $[\![26,18,5]\!]_5$\\
\( 26 \)    &\([\![ 26 , 0 , 14 ]\!]_5 \)   	 &\([\![ 17 , 9 , 5 ]\!]_5 \) && shortening $[\![26,18,5]\!]_5$\\
\( 28 \)    &\([\![ 26 , 2 , 13 ]\!]_5 \)   	 &\([\![ 18 , 10, 5 ]\!]_5 \) && shortening $[\![26,18,5]\!]_5$\\
\( 30 \)    &\([\![ 30 , 0 , 16 ]\!]_5 \)   	 &\([\![ 19 , 11, 5 ]\!]_5 \) && shortening $[\![26,18,5]\!]_5$\\
\( 32 \)    &\([\![ 30 , 2 , 15 ]\!]_5 \)   	 &\([\![ 20 , 12, 5 ]\!]_5 \) && shortening $[\![26,18,5]\!]_5$\\
\( 34 \)    &\([\![ 34 , 0 , 18 ]\!]_5 \)   	 &\([\![ 21 , 13, 5 ]\!]_5 \) && shortening $[\![26,18,5]\!]_5$\\
\( 36 \)    &\([\![ 34 , 2 , 17 ]\!]_5 \)   	 &\([\![ 22 , 14, 5 ]\!]_5 \) && shortening $[\![26,18,5]\!]_5$\\
\( 38 \)    &\([\![ 38 , 0 , 20 ]\!]_5 \)   	 &\([\![ 23 , 15, 5 ]\!]_5 \) && shortening $[\![26,18,5]\!]_5$\\
\( 40 \)    &\([\![ 37 , 3 , 18 ]\!]_5 \)   	 &\([\![ 24 , 16, 5 ]\!]_5 \) && shortening $[\![26,18,5]\!]_5$\\
\( 42 \)    &\([\![ 42 , 0 , 22 ]\!]_5 \)   	 &\([\![ 26 , 16, 6 ]\!]_5 \) && Grassl/Rötteler I\\
\( 44 \)    &\([\![ 41 , 3 , 20 ]\!]_5 \)   	 &\([\![ 26 , 18, 5 ]\!]_5 \) && Grassl/Rötteler I\\
\( 46 \)    &\([\![ 46 , 0 , 24 ]\!]_5 \)   	 &\([\![ 26 , 20, 4 ]\!]_5 \) && Grassl/Rötteler I\\
\( 48 \)    &\([\![ 45 , 3 , 22 ]\!]_5 \)   	 &\([\![ 26 , 22, 3 ]\!]_5 \) && single-error\\
\hline
\end{tabular}
\caption{\label{tab:bounds_D=5} Upper and lower bounds for the 
highest distance in QMDS families of local dimension~\(D = 5\).}
\end{table}

%%%%%%%%%%%%%%%
\bibliographystyle{apsrev4-1}
\bibliography{current_bib}

\begin{thebibliography}{72}%
\makeatletter
\providecommand \@ifxundefined [1]{%
 \@ifx{#1\undefined}
}%
\providecommand \@ifnum [1]{%
 \ifnum #1\expandafter \@firstoftwo
 \else \expandafter \@secondoftwo
 \fi
}%
\providecommand \@ifx [1]{%
 \ifx #1\expandafter \@firstoftwo
 \else \expandafter \@secondoftwo
 \fi
}%
\providecommand \natexlab [1]{#1}%
\providecommand \enquote  [1]{``#1''}%
\providecommand \bibnamefont  [1]{#1}%
\providecommand \bibfnamefont [1]{#1}%
\providecommand \citenamefont [1]{#1}%
\providecommand \href@noop [0]{\@secondoftwo}%
\providecommand \href [0]{\begingroup \@sanitize@url \@href}%
\providecommand \@href[1]{\@@startlink{#1}\@@href}%
\providecommand \@@href[1]{\endgroup#1\@@endlink}%
\providecommand \@sanitize@url [0]{\catcode `\\12\catcode `\$12\catcode
  `\&12\catcode `\#12\catcode `\^12\catcode `\_12\catcode `\%12\relax}%
\providecommand \@@startlink[1]{}%
\providecommand \@@endlink[0]{}%
\providecommand \url  [0]{\begingroup\@sanitize@url \@url }%
\providecommand \@url [1]{\endgroup\@href {#1}{\urlprefix }}%
\providecommand \urlprefix  [0]{URL }%
\providecommand \Eprint [0]{\href }%
\providecommand \doibase [0]{http://doi.org/}%
\providecommand \selectlanguage [0]{\@gobble}%
\providecommand \bibinfo  [0]{\@secondoftwo}%
\providecommand \bibfield  [0]{\@secondoftwo}%
\providecommand \translation [1]{[#1]}%
\providecommand \BibitemOpen [0]{}%
\providecommand \bibitemStop [0]{}%
\providecommand \bibitemNoStop [0]{.\EOS\space}%
\providecommand \EOS [0]{\spacefactor3000\relax}%
\providecommand \BibitemShut  [1]{\csname bibitem#1\endcsname}%
\let\auto@bib@innerbib\@empty
%</preamble>
\bibitem [{\citenamefont {Shor}(1995)}]{PhysRevA.52.R2493}%
  \BibitemOpen
  \bibfield  {author} {\bibinfo {author} {\bibfnamefont {P.~W.}\ \bibnamefont
  {Shor}},\ }\bibfield  {title} {\enquote {\bibinfo {title} {{Scheme for
  reducing decoherence in quantum computer memory}},}\ }\href {\doibase 10.1103/PhysRevA.52.R2493} {\bibfield  {journal} {\bibinfo  {journal} {Phys.
  Rev. A}\ }\textbf {\bibinfo {volume} {52}},\ \bibinfo {pages} {2493(R)}
  (\bibinfo {year} {1995})}\BibitemShut {NoStop}%
\bibitem [{\citenamefont {Steane}(1996)}]{PhysRevLett.77.793}%
  \BibitemOpen
  \bibfield  {author} {\bibinfo {author} {\bibfnamefont {A.~M.}\ \bibnamefont
  {Steane}},\ }\bibfield  {title} {\enquote {\bibinfo {title} {{Error
  Correcting Codes in Quantum Theory}},}\ }\href {\doibase 10.1103/PhysRevLett.77.793} {\bibfield  {journal} {\bibinfo  {journal} {Phys.
  Rev. Lett.}\ }\textbf {\bibinfo {volume} {77}},\ \bibinfo {pages} {793--797}
  (\bibinfo {year} {1996})}\BibitemShut {NoStop}%
\bibitem [{\citenamefont {Knill}\ and\ \citenamefont
  {Laflamme}(1997)}]{PhysRevA.55.900}%
  \BibitemOpen
  \bibfield  {author} {\bibinfo {author} {\bibfnamefont {E.}~\bibnamefont
  {Knill}}\ and\ \bibinfo {author} {\bibfnamefont {R.}~\bibnamefont
  {Laflamme}},\ }\bibfield  {title} {\enquote {\bibinfo {title} {{Theory of
  quantum error-correcting codes}},}\ }\href {\doibase 10.1103/PhysRevA.55.900}
  {\bibfield  {journal} {\bibinfo  {journal} {Phys. Rev. A}\ }\textbf {\bibinfo
  {volume} {55}},\ \bibinfo {pages} {900--911} (\bibinfo {year}
  {1997})}\BibitemShut {NoStop}%
\bibitem [{\citenamefont {Bennett}\ \emph {et~al.}(1996)\citenamefont
  {Bennett}, \citenamefont {DiVincenzo}, \citenamefont {Smolin},\ and\
  \citenamefont {Wootters}}]{PhysRevA.54.3824}%
  \BibitemOpen
  \bibfield  {author} {\bibinfo {author} {\bibfnamefont {C.~H.}\ \bibnamefont
  {Bennett}}, \bibinfo {author} {\bibfnamefont {D.~P.}\ \bibnamefont
  {DiVincenzo}}, \bibinfo {author} {\bibfnamefont {J.~A.}\ \bibnamefont
  {Smolin}}, \ and\ \bibinfo {author} {\bibfnamefont {W.~K.}\ \bibnamefont
  {Wootters}},\ }\bibfield  {title} {\enquote {\bibinfo {title} {{Mixed-state
  entanglement and quantum error correction}},}\ }\href {\doibase 10.1103/PhysRevA.54.3824} {\bibfield  {journal} {\bibinfo  {journal} {Phys.
  Rev. A}\ }\textbf {\bibinfo {volume} {54}},\ \bibinfo {pages} {3824}
  (\bibinfo {year} {1996})}\BibitemShut {NoStop}%
\bibitem [{\citenamefont {Schumacher}\ and\ \citenamefont
  {Nielsen}(1996)}]{PhysRevA.54.2629}%
  \BibitemOpen
  \bibfield  {author} {\bibinfo {author} {\bibfnamefont {B.}~\bibnamefont
  {Schumacher}}\ and\ \bibinfo {author} {\bibfnamefont {M.~A.}\ \bibnamefont
  {Nielsen}},\ }\bibfield  {title} {\enquote {\bibinfo {title} {{Quantum data
  processing and error correction}},}\ }\href
  {https://doi.org/10.1103/PhysRevA.54.2629} {\bibfield  {journal}
  {\bibinfo  {journal} {Phys. Rev. A}\ }\textbf {\bibinfo {volume} {54}},\
  \bibinfo {pages} {2629--2635} (\bibinfo {year} {1996})}\BibitemShut {NoStop}%
\bibitem [{\citenamefont {Calderbank}\ \emph {et~al.}(1998)\citenamefont
  {Calderbank}, \citenamefont {Rains}, \citenamefont {Shor},\ and\
  \citenamefont {Sloane}}]{681315}%
  \BibitemOpen
  \bibfield  {author} {\bibinfo {author} {\bibfnamefont {A.~R.}\ \bibnamefont
  {Calderbank}}, \bibinfo {author} {\bibfnamefont {E.~M.}\ \bibnamefont
  {Rains}}, \bibinfo {author} {\bibfnamefont {P.~M.}\ \bibnamefont {Shor}}, \
  and\ \bibinfo {author} {\bibfnamefont {N.~J.~A.}\ \bibnamefont {Sloane}},\
  }\bibfield  {title} {\enquote {\bibinfo {title} {{Quantum error correction
  via codes over GF(4)}},}\ }\href {\doibase 10.1109/18.681315} {\bibfield
  {journal} {\bibinfo  {journal} {IEEE Trans. Inf. Theory}\ }\textbf {\bibinfo
  {volume} {44}},\ \bibinfo {pages} {1369} (\bibinfo {year}
  {1998})}\BibitemShut {NoStop}%
\bibitem [{\citenamefont {Dennis}\ \emph {et~al.}(2002)\citenamefont {Dennis},
  \citenamefont {Kitaev}, \citenamefont {Landahl},\ and\ \citenamefont
  {Preskill}}]{doi:10.1063/1.1499754}%
  \BibitemOpen
  \bibfield  {author} {\bibinfo {author} {{\bibfnamefont\rule{0pt}{31.5mm}E.}~\bibnamefont
  {Dennis}}, \bibinfo {author} {\bibfnamefont {A.}~\bibnamefont {Kitaev}},
  \bibinfo {author} {\bibfnamefont {A.}~\bibnamefont {Landahl}}, \ and\
  \bibinfo {author} {\bibfnamefont {J.}~\bibnamefont {Preskill}},\ }\bibfield
  {title} {\enquote {\bibinfo {title} {{Topological quantum memory}},}\ }\href
  {\doibase doi.org/10.1063/1.1499754} {\bibfield  {journal} {\bibinfo
  {journal} {J. Math. Phys.}\ }\textbf {\bibinfo {volume} {43}},\ \bibinfo
  {pages} {4452--4505} (\bibinfo {year} {2002})}\BibitemShut {NoStop}%
\bibitem [{\citenamefont {Bombin}\ and\ \citenamefont
  {Martin-Delgado}(2006)}]{PhysRevLett.97.180501}%
  \BibitemOpen
  \bibfield  {author} {\bibinfo {author} {\bibfnamefont {H.}~\bibnamefont
  {Bombin}}\ and\ \bibinfo {author} {\bibfnamefont {M.~A.}\ \bibnamefont
  {Martin-Delgado}},\ }\bibfield  {title} {\enquote {\bibinfo {title}
  {{Topological Quantum Distillation}},}\ }\href {\doibase 10.1103/PhysRevLett.97.180501} {\bibfield  {journal} {\bibinfo  {journal}
  {Phys. Rev. Lett.}\ }\textbf {\bibinfo {volume} {97}},\ \bibinfo {pages}
  {180501} (\bibinfo {year} {2006})}\BibitemShut {NoStop}%
\bibitem [{\citenamefont {Tillich}\ and\ \citenamefont
  {Z{\'e}mor}(2014)}]{6671468}%
  \BibitemOpen
  \bibfield  {author} {\bibinfo {author} {\bibfnamefont {J.}~\bibnamefont
  {Tillich}}\ and\ \bibinfo {author} {\bibfnamefont {G.}~\bibnamefont
  {Z{\'e}mor}},\ }\bibfield  {title} {\enquote {\bibinfo {title} {{Quantum LDPC
  Codes With Positive Rate and Minimum Distance Proportional to the Square Root
  of the Blocklength}},}\ }\href {\doibase 10.1109/TIT.2013.2292061} {\bibfield
   {journal} {\bibinfo  {journal} {IEEE Trans. Inf. Theory}\ }\textbf {\bibinfo
  {volume} {60}},\ \bibinfo {pages} {1193--1202} (\bibinfo {year}
  {2014})}\BibitemShut {NoStop}%
\bibitem [{\citenamefont {Hastings}\ and\ \citenamefont
  {Bravyi}(2014)}]{Hastings2014}%
  \BibitemOpen
  \bibfield  {author} {\bibinfo {author} {\bibfnamefont {M.~B.}\ \bibnamefont
  {Hastings}}\ and\ \bibinfo {author} {\bibfnamefont {S.}~\bibnamefont
  {Bravyi}},\ }\bibfield  {title} {\enquote {\bibinfo {title} {{Homological
  Product Codes}},}\ }in\ \href {\doibase 10.1145/2591796.2591870} {\emph
  {\bibinfo {booktitle} {Proc. of the 46th ACM Symposium on Theory of
  Computing}}}\ (\bibinfo {year} {2014})\ pp.\ \bibinfo {pages}
  {273--282}\BibitemShut {NoStop}%
\bibitem [{\citenamefont {Terhal}(2015)}]{RevModPhys.87.307}%
  \BibitemOpen
  \bibfield  {author} {\bibinfo {author} {\bibfnamefont {B.~M.}\ \bibnamefont
  {Terhal}},\ }\bibfield  {title} {\enquote {\bibinfo {title} {{Quantum error
  correction for quantum memories}},}\ }\href {\doibase 10.1103/RevModPhys.87.307} {\bibfield  {journal} {\bibinfo  {journal} {Rev.
  Mod. Phys.}\ }\textbf {\bibinfo {volume} {87}},\ \bibinfo {pages} {307--346}
  (\bibinfo {year} {2015})}\BibitemShut {NoStop}%
\bibitem [{\citenamefont {Campbell}\ \emph {et~al.}(2017)\citenamefont
  {Campbell}, \citenamefont {Terhal},\ and\ \citenamefont
  {Christophe}}]{Campbell2017}%
  \BibitemOpen
  \bibfield  {author} {\bibinfo {author} {\bibfnamefont {E.~T.}\ \bibnamefont
  {Campbell}}, \bibinfo {author} {\bibfnamefont {B.~M.}\ \bibnamefont
  {Terhal}}, \ and\ \bibinfo {author} {\bibfnamefont {V.}~\bibnamefont
  {Christophe}},\ }\bibfield  {title} {\enquote {\bibinfo {title} {{Roads
  towards fault-tolerant universal quantum computation}},}\ }\href {\doibase 10.1038/nature23460} {\bibfield  {journal} {\bibinfo  {journal} {Nature}\
  }\textbf {\bibinfo {volume} {549}},\ \bibinfo {pages} {172} (\bibinfo {year}
  {2017})}\BibitemShut {NoStop}%
\bibitem [{\citenamefont {{Ashikhmin}}\ and\ \citenamefont
  {{Litsyu}}(1999)}]{761270}%
  \BibitemOpen
  \bibfield  {author} {\bibinfo {author} {\bibfnamefont {A.}~\bibnamefont
  {{Ashikhmin}}}\ and\ \bibinfo {author} {\bibfnamefont {S.}~\bibnamefont
  {{Litsyu}}},\ }\bibfield  {title} {\enquote {\bibinfo {title} {{Upper bounds
  on the size of quantum codes}},}\ }\href {\doibase 10.1109/18.761270}
  {\bibfield  {journal} {\bibinfo  {journal} {IEEE Trans. Inf. Theory}\
  }\textbf {\bibinfo {volume} {45}},\ \bibinfo {pages} {1206--1215} (\bibinfo
  {year} {1999})}\BibitemShut {NoStop}%
\bibitem [{\citenamefont {Chandra}\ \emph {et~al.}(2017)\citenamefont
  {Chandra}, \citenamefont {Babar}, \citenamefont {Nguyen}, \citenamefont
  {Alanis}, \citenamefont {Botsinis}, \citenamefont {Ng},\ and\ \citenamefont
  {Hanzo}}]{7950914}%
  \BibitemOpen
  \bibfield  {author} {\bibinfo {author} {\bibfnamefont {D.}~\bibnamefont
  {Chandra}}, \bibinfo {author} {\bibfnamefont {Z.}~\bibnamefont {Babar}},
  \bibinfo {author} {\bibfnamefont {H.~V.}\ \bibnamefont {Nguyen}}, \bibinfo
  {author} {\bibfnamefont {D.}~\bibnamefont {Alanis}}, \bibinfo {author}
  {\bibfnamefont {P.}~\bibnamefont {Botsinis}}, \bibinfo {author}
  {\bibfnamefont {S.~X.}\ \bibnamefont {Ng}}, \ and\ \bibinfo {author}
  {\bibfnamefont {L.}~\bibnamefont {Hanzo}},\ }\bibfield  {title} {\enquote
  {\bibinfo {title} {{Quantum Coding Bounds and a Closed-Form Approximation of
  the Minimum Distance Versus Quantum Coding Rate}},}\ }\href {\doibase 10.1109/ACCESS.2017.2716367} {\bibfield  {journal} {\bibinfo  {journal} {IEEE
  Access}\ }\textbf {\bibinfo {volume} {5}},\ \bibinfo {pages} {11557--11581}
  (\bibinfo {year} {2017})}\BibitemShut {NoStop}%
\bibitem [{\citenamefont {Cerf}\ and\ \citenamefont
  {Cleve}(1997)}]{PhysRevA.56.1721}%
  \BibitemOpen
  \bibfield  {author} {\bibinfo {author} {\bibfnamefont {N.~J.}\ \bibnamefont
  {Cerf}}\ and\ \bibinfo {author} {\bibfnamefont {R.}~\bibnamefont {Cleve}},\
  }\bibfield  {title} {\enquote {\bibinfo {title} {{Information-theoretic
  interpretation of quantum error-correcting codes}},}\ }\href {\doibase 10.1103/PhysRevA.56.1721} {\bibfield  {journal} {\bibinfo  {journal} {Phys.
  Rev. A}\ }\textbf {\bibinfo {volume} {56}},\ \bibinfo {pages} {1721--1732}
  (\bibinfo {year} {1997})}\BibitemShut {NoStop}%
\bibitem [{\citenamefont {Gottesman}(2004)}]{lecture_Gottesman}%
  \BibitemOpen
  \bibfield  {author} {\bibinfo {author} {\bibfnamefont {D.}~\bibnamefont
  {Gottesman}},\ }\href@noop {} {\enquote {\bibinfo {title} {{Lecture Notes
  CO639}},}\ }\bibinfo {howpublished} {available online at
  \url{www.perimeterinstitute.ca/personal/dgottesman/CO639-2004/}} (\bibinfo
  {year} {2004})\BibitemShut {NoStop}%
\bibitem [{\citenamefont {Rains}(1999{\natexlab{a}})}]{782103}%
  \BibitemOpen
  \bibfield  {author} {\bibinfo {author} {\bibfnamefont {E.~M.}\ \bibnamefont
  {Rains}},\ }\bibfield  {title} {\enquote {\bibinfo {title} {{Nonbinary
  quantum codes}},}\ }\href {\doibase 10.1109/18.782103} {\bibfield  {journal}
  {\bibinfo  {journal} {IEEE Trans. Inf. Theory}\ }\textbf {\bibinfo {volume}
  {45}},\ \bibinfo {pages} {1827} (\bibinfo {year}
  {1999}{\natexlab{a}})}\BibitemShut {NoStop}%
\bibitem [{\citenamefont {D{\"u}r}\ \emph {et~al.}(2000)\citenamefont
  {D{\"u}r}, \citenamefont {Vidal},\ and\ \citenamefont
  {Cirac}}]{PhysRevA.62.062314}%
  \BibitemOpen
  \bibfield  {author} {\bibinfo {author} {\bibfnamefont {W.}~\bibnamefont
  {D{\"u}r}}, \bibinfo {author} {\bibfnamefont {G.}~\bibnamefont {Vidal}}, \
  and\ \bibinfo {author} {\bibfnamefont {J.~I.}\ \bibnamefont {Cirac}},\
  }\bibfield  {title} {\enquote {\bibinfo {title} {{Three qubits can be
  entangled in two inequivalent ways}},}\ }\href {\doibase 10.1103/PhysRevA.62.062314} {\bibfield  {journal} {\bibinfo  {journal} {Phys.
  Rev. A}\ }\textbf {\bibinfo {volume} {62}},\ \bibinfo {pages} {062314}
  (\bibinfo {year} {2000})}\BibitemShut {NoStop}%
\bibitem [{\citenamefont {Walter}\ \emph {et~al.}(2017)\citenamefont {Walter},
  \citenamefont {Gross},\ and\ \citenamefont {Eisert}}]{1612.02437v2}%
  \BibitemOpen
  \bibfield  {author} {\bibinfo {author} {\bibfnamefont {M.}~\bibnamefont
  {Walter}}, \bibinfo {author} {\bibfnamefont {D.}~\bibnamefont {Gross}}, \
  and\ \bibinfo {author} {\bibfnamefont {J.}~\bibnamefont {Eisert}},\
  }\href@noop {} {\enquote {\bibinfo {title} {{Multi-par\-tite entanglement}},}\
  } (\bibinfo {year} {2017}),\ \Eprint {http://arxiv.org/abs/1612.02437}
  {arXiv:1612.02437}\BibitemShut{NoStop}%
\bibitem [{\citenamefont {Bengtsson}\ and\ \citenamefont
  {{\.Z}yczkowski}(2017)}]{Bengtsson_Zyczkowski_2017}%
  \BibitemOpen
  \bibfield  {author} {\bibinfo {author} {\bibfnamefont {I.}~\bibnamefont
  {Bengtsson}}\ and\ \bibinfo {author} {\bibfnamefont {K.}~\bibnamefont
  {{\.Z}yczkowski}},\ }\href {\doibase 10.1017/CBO9780511535048} {\emph
  {\bibinfo {title} {{Geometry of Quantum States: An Introduction to Quantum
  Entanglement}}}},\ \bibinfo {edition} {2nd}\ ed.\ (\bibinfo  {publisher}
  {Cambridge University Press},\ \bibinfo {year} {2017})\BibitemShut {NoStop}%
\bibitem [{\citenamefont {Cubitt}\ \emph {et~al.}(2008)\citenamefont {Cubitt},
  \citenamefont {Montanaro},\ and\ \citenamefont
  {Winter}}]{doi:10.1063/1.2862998}%
  \BibitemOpen
  \bibfield  {author} {\bibinfo {author} {\bibfnamefont {T.}~\bibnamefont
  {Cubitt}}, \bibinfo {author} {\bibfnamefont {A.}~\bibnamefont {Montanaro}}, \
  and\ \bibinfo {author} {\bibfnamefont {A.}~\bibnamefont {Winter}},\
  }\bibfield  {title} {\enquote {\bibinfo {title} {{On the dimension of
  subspaces with bounded Schmidt rank}},}\ }\href {\doibase 10.1063/1.2862998}
  {\bibfield  {journal} {\bibinfo  {journal} {J. Math. Phys.}\
  }\textbf {\bibinfo {volume} {49}},\ \bibinfo {pages} {022107} (\bibinfo
  {year} {2008})}\BibitemShut {NoStop}%
\bibitem [{\citenamefont {Johnston}(2013)}]{PhysRevA.87.064302}%
  \BibitemOpen
  \bibfield  {author} {\bibinfo {author} {\bibfnamefont {N.}~\bibnamefont
  {Johnston}},\ }\bibfield  {title} {\enquote {\bibinfo {title}
  {{Non-positive-partial-transpose sub\-spaces can be as large as any entangled
  sub\-space}},}\ }\href {\doibase 10.1103/PhysRevA.87.064302} {\bibfield
  {journal} {\bibinfo  {journal} {Phys. Rev. A}\ }\textbf {\bibinfo {volume}
  {87}},\ \bibinfo {pages} {064302} (\bibinfo {year} {2013})}\BibitemShut
  {NoStop}%
\bibitem [{\citenamefont {Walgate}\ and\ \citenamefont
  {Scott}(2008)}]{1751-8121-41-37-375305}%
  \BibitemOpen
  \bibfield  {author} {\bibinfo {author} {\bibfnamefont {J.}~\bibnamefont
  {Walgate}}\ and\ \bibinfo {author} {\bibfnamefont {A.~J.}\ \bibnamefont
  {Scott}},\ }\bibfield  {title} {\enquote {\bibinfo {title} {{Generic local
  distinguishability and completely entangled subspaces}},}\ }\href {\doibase 10.1088/1751-8113/41/37/375305} {\bibfield  {journal} {\bibinfo  {journal}
  {J. Phys. A: Math. Theor.}\ }\textbf {\bibinfo {volume} {41}},\ \bibinfo
  {pages} {375305} (\bibinfo {year} {2008})}\BibitemShut {NoStop}%
\bibitem [{\citenamefont {Sengupta}\ \emph {et~al.}(2014)\citenamefont
  {Sengupta}, \citenamefont {Arvind},\ and\ \citenamefont
  {Singh}}]{PhysRevA.90.062323}%
  \BibitemOpen
  \bibfield  {author} {\bibinfo {author} {\bibfnamefont {R.}~\bibnamefont
  {Sengupta}}, \bibinfo {author} {\bibnamefont {Arvind}}, \ and\ \bibinfo
  {author} {\bibfnamefont {A.~I.}\ \bibnamefont {Singh}},\ }\bibfield  {title}
  {\enquote {\bibinfo {title} {{Entanglement properties of positive operators
  with ranges in completely entangled subspaces}},}\ }\href {\doibase 10.1103/PhysRevA.90.062323} {\bibfield  {journal} {\bibinfo  {journal} {Phys.
  Rev. A}\ }\textbf {\bibinfo {volume} {90}},\ \bibinfo {pages} {062323}
  (\bibinfo {year} {2014})}\BibitemShut {NoStop}%
\bibitem [{\citenamefont {Demianowicz}\ and\ \citenamefont
  {Augusiak}(2018)}]{PhysRevA.98.012313}%
  \BibitemOpen
  \bibfield  {author} {\bibinfo {author} {\bibfnamefont {M.}~\bibnamefont
  {Demianowicz}}\ and\ \bibinfo {author} {\bibfnamefont {R.}~\bibnamefont
  {Augusiak}},\ }\bibfield  {title} {\enquote {\bibinfo {title} {{From
  unextendible product bases to genuinely entangled subspaces}},}\ }\href
  {\doibase 10.1103/PhysRevA.98.012313} {\bibfield  {journal} {\bibinfo
  {journal} {Phys. Rev. A}\ }\textbf {\bibinfo {volume} {98}},\ \bibinfo
  {pages} {012313} (\bibinfo {year} {2018})}\BibitemShut {NoStop}%
\bibitem [{\citenamefont {Demianowicz}\ and\ \citenamefont
  {Augusiak}(2019)}]{PhysRevA.100.062318}%
  \BibitemOpen
  \bibfield  {author} {\bibinfo {author} {\bibfnamefont {M.}~\bibnamefont
  {Demianowicz}}\ and\ \bibinfo {author} {\bibfnamefont {R.}~\bibnamefont
  {Augusiak}},\ }\bibfield  {title} {\enquote {\bibinfo {title} {Entanglement
  of genuinely entangled subspaces and states: Exact, approximate, and
  numerical results},}\ }\href {\doibase 10.1103/PhysRevA.100.062318}
  {\bibfield  {journal} {\bibinfo  {journal} {Phys. Rev. A}\ }\textbf {\bibinfo
  {volume} {100}},\ \bibinfo {pages} {062318} (\bibinfo {year}
  {2019})}\BibitemShut {NoStop}%
\bibitem [{\citenamefont {Demianowicz}\ and\ \citenamefont
  {Augusiak}(2020)}]{demianowicz2020approach}%
  \BibitemOpen
  \bibfield  {author} {\bibinfo {author} {\bibfnamefont {M.}~\bibnamefont
  {Demianowicz}}\ and\ \bibinfo {author} {\bibfnamefont {R.}~\bibnamefont
  {Augusiak}},\ }\bibfield  {title} {\enquote {\bibinfo {title} {An approach to
  constructing genuinely entangled subspaces of maximal dimension},}\ }\href
  {\doibase 10.1007/s11128-020-02688-4} {\bibfield  {journal} {\bibinfo
  {journal} {Quantum Inf. Process.}\ }\textbf {\bibinfo {volume} {19}},\
  \bibinfo {pages} {199} (\bibinfo {year} {2020})}\BibitemShut {NoStop}%
\bibitem [{\citenamefont {Gour}\ and\ \citenamefont
  {Wallach}(2007)}]{PhysRevA.76.042309}%
  \BibitemOpen
  \bibfield  {author} {\bibinfo {author} {\bibfnamefont {G.}~\bibnamefont
  {Gour}}\ and\ \bibinfo {author} {\bibfnamefont {N.~R.}\ \bibnamefont
  {Wallach}},\ }\bibfield  {title} {\enquote {\bibinfo {title} {{Entanglement
  of subspaces and error-correcting codes}},}\ }\href {\doibase 10.1103/PhysRevA.76.042309} {\bibfield  {journal} {\bibinfo  {journal} {Phys.
  Rev. A}\ }\textbf {\bibinfo {volume} {76}},\ \bibinfo {pages} {042309}
  (\bibinfo {year} {2007})}\BibitemShut {NoStop}%
\bibitem [{\citenamefont {Grassl}\ \emph {et~al.}(2004)\citenamefont {Grassl},
  \citenamefont {Beth},\ and\ \citenamefont
  {R{\"o}tteler}}]{quant-ph/0312164v1}%
  \BibitemOpen
  \bibfield  {author} {\bibinfo {author} {\bibfnamefont {M.}~\bibnamefont
  {Grassl}}, \bibinfo {author} {\bibfnamefont {T.}~\bibnamefont {Beth}}, \ and\
  \bibinfo {author} {\bibfnamefont {M.}~\bibnamefont {R{\"o}tteler}},\
  }\bibfield  {title} {\enquote {\bibinfo {title} {{On optimal quantum
  codes}},}\ }\href {http://doi.org/10.1142/S0219749904000079} {\bibfield
  {journal} {\bibinfo  {journal} {Int. J. Quant. Inf.}\ }\textbf {\bibinfo
  {volume} {2}},\ \bibinfo {pages} {55} (\bibinfo {year} {2004})}\BibitemShut
  {NoStop}%
\bibitem [{\citenamefont {Grassl}\ and\ \citenamefont
  {R{\"o}tteler}(2015)}]{7282626}%
  \BibitemOpen
  \bibfield  {author} {\bibinfo {author} {\bibfnamefont {M.}~\bibnamefont
  {Grassl}}\ and\ \bibinfo {author} {\bibfnamefont {M.}~\bibnamefont
  {R{\"o}tteler}},\ }\bibfield  {title} {\enquote {\bibinfo {title} {{Quantum
  MDS codes over small fields}},}\ }in\ \href {\doibase 10.1109/ISIT.2015.7282626} {\emph {\bibinfo {booktitle} {{IEEE Int. Symp.
  Inf. Theory (ISIT 2015)}}}}\ (\bibinfo {year} {2015})\ pp.\ \bibinfo {pages}
  {1104--1108}\BibitemShut {NoStop}%
\bibitem [{\citenamefont {Guardia}(2011)}]{5961827}%
  \BibitemOpen
  \bibfield  {author} {\bibinfo {author} {\bibfnamefont {G.~G.~L.}\
  \bibnamefont {Guardia}},\ }\bibfield  {title} {\enquote {\bibinfo {title}
  {{New Quantum MDS Codes}},}\ }\href {\doibase 10.1109/TIT.2011.2159039}
  {\bibfield  {journal} {\bibinfo  {journal} {IEEE Trans. Inf. Theory}\
  }\textbf {\bibinfo {volume} {57}},\ \bibinfo {pages} {5551--5554} (\bibinfo
  {year} {2011})}\BibitemShut {NoStop}%
\bibitem [{\citenamefont {Helwig}(2014)}]{Helwig_Thesis}%
  \BibitemOpen
  \bibfield  {author} {\bibinfo {author} {\bibfnamefont {W.}~\bibnamefont
  {Helwig}},\ }\emph {\bibinfo {title} {{Multipartite Entanglement:
  Transformations, Quantum Secret Sharing, Quantum Error Correction}}},\ \href
  {http://hdl.handle.net/1807/44114} {Ph.D. thesis},\ \bibinfo  {school}
  {University of Toronto} (\bibinfo {year} {2014})\BibitemShut {NoStop}%
\bibitem [{\citenamefont {Joyner}\ and\ \citenamefont
  {Kim}(2011)}]{Joyner:2011:SUP:2073646}%
  \BibitemOpen
  \bibfield  {author} {\bibinfo {author} {\bibfnamefont {D.}~\bibnamefont
  {Joyner}}\ and\ \bibinfo {author} {\bibfnamefont {J.-L.}\ \bibnamefont
  {Kim}},\ }\href {\doibase 10.1007/978-0-8176-8256-9} {\emph {\bibinfo {title}
  {{Selected Unsolved Problems in Coding Theory (Applied and Numerical Harmonic
  Analysis)}}}},\ \bibinfo {edition} {1st}\ ed.\ (\bibinfo  {publisher}
  {Birkh{\"a}user Basel},\ \bibinfo {year} {2011})\BibitemShut {NoStop}%
\bibitem [{\citenamefont {Ball}(2012)}]{Ball2012}%
  \BibitemOpen
  \bibfield  {author} {\bibinfo {author} {\bibfnamefont {S.}~\bibnamefont
  {Ball}},\ }\bibfield  {title} {\enquote {\bibinfo {title} {On sets of vectors
  of a finite vector space in which every subset of basis size is a basis},}\
  }\href {\doibase 10.4171/JEMS/316} {\bibfield  {journal} {\bibinfo  {journal}
  {J. Eur. Math. Soc.}\ }\textbf {\bibinfo {volume} {14}},\ \bibinfo {pages}
  {733--748} (\bibinfo {year} {2012})}\BibitemShut {NoStop}%
\bibitem [{\citenamefont {Eltschka}\ \emph {et~al.}(2018)\citenamefont
  {Eltschka}, \citenamefont {Huber}, \citenamefont {G{\"u}hne},\ and\
  \citenamefont {Siewert}}]{PhysRevA.98.052317}%
  \BibitemOpen
  \bibfield  {author} {\bibinfo {author} {\bibfnamefont {C.}~\bibnamefont
  {Eltschka}}, \bibinfo {author} {\bibfnamefont {F.}~\bibnamefont {Huber}},
  \bibinfo {author} {\bibfnamefont {O.}~\bibnamefont {G{\"u}hne}}, \ and\
  \bibinfo {author} {\bibfnamefont {J.}~\bibnamefont {Siewert}},\ }\bibfield
  {title} {\enquote {\bibinfo {title} {{Exponentially many entanglement and
  correlation constraints for multipartite quantum states}},}\ }\href {\doibase 10.1103/PhysRevA.98.052317} {\bibfield  {journal} {\bibinfo  {journal} {Phys.
  Rev. A}\ }\textbf {\bibinfo {volume} {98}},\ \bibinfo {pages} {052317}
  (\bibinfo {year} {2018})}\BibitemShut {NoStop}%
\bibitem [{\citenamefont {Knill}\ \emph {et~al.}(2000)\citenamefont {Knill},
  \citenamefont {Laflamme},\ and\ \citenamefont {Viola}}]{PhysRevLett.84.2525}%
  \BibitemOpen
  \bibfield  {author} {\bibinfo {author} {\bibfnamefont {E.}~\bibnamefont
  {Knill}}, \bibinfo {author} {\bibfnamefont {R.}~\bibnamefont {Laflamme}}, \
  and\ \bibinfo {author} {\bibfnamefont {L.}~\bibnamefont {Viola}},\ }\bibfield
   {title} {\enquote {\bibinfo {title} {{Theory of Quantum Error Correction for
  General Noise}},}\ }\href {\doibase 10.1103/PhysRevLett.84.2525} {\bibfield
  {journal} {\bibinfo  {journal} {Phys. Rev. Lett.}\ }\textbf {\bibinfo
  {volume} {84}},\ \bibinfo {pages} {2525--2528} (\bibinfo {year}
  {2000})}\BibitemShut {NoStop}%
%\bibitem [{\citenamefont {Gottesman}(2002)}]{quant-ph/0004072v1}%
%  \BibitemOpen
%  \bibfield  {author} {\bibinfo {author} {\bibfnamefont {D.}~\bibnamefont
%  {Gottesman}},\ }\enquote {\bibinfo {title} {Quantum computation: A grand
%  mathematical challenge for the twenty-first century and the millennium},}\ \
%  (\bibinfo  {publisher} {American Mathematical Society},\ \bibinfo {year}
%  {2002})\ Chap.\ \bibinfo {chapter} {{An Introduction to Quantum Error
%  Correction}}, pp.\ \bibinfo {pages} {221--235},\ \Eprint
%  {http://arxiv.org/abs/quant-ph/0004072} {arXiv:quant-ph/0004072}\BibitemShut
%  {NoStop}%
\bibitem [{\citenamefont {Gottesman}(2002)}]{quant-ph/0004072v1}%
  \BibitemOpen
  \bibfield  {author} {\bibinfo {author} {\bibfnamefont {D.}~\bibnamefont
  {Gottesman}},\ }\enquote {\bibinfo {title} {An Introduction to Quantum Error
  Correction},}\ \
  in\ \emph{\bibinfo {booktitle} {{Quantum computation: A grand
  mathematical challenge for the twenty-first century and the
  millennium},}}\ {\bibinfo{editor} {ed. S. J. Lomonaco, Jr.,}}\ 
  pp.\ \bibinfo {pages} {221--235},\ 
  (\bibinfo  {publisher} {American Mathematical Society},\ \bibinfo {year}
  {2002})\ \Eprint
  {http://arxiv.org/abs/quant-ph/0004072} {arXiv:quant-ph/0004072}\BibitemShut
  {NoStop}%
\bibitem [{\citenamefont {Cleve}\ \emph {et~al.}(1999)\citenamefont {Cleve},
  \citenamefont {Gottesman},\ and\ \citenamefont {Lo}}]{PhysRevLett.83.648}%
  \BibitemOpen
  \bibfield  {author} {\bibinfo {author} {\bibfnamefont {R.}~\bibnamefont
  {Cleve}}, \bibinfo {author} {\bibfnamefont {D.}~\bibnamefont {Gottesman}}, \
  and\ \bibinfo {author} {\bibfnamefont {H.-K.}\ \bibnamefont {Lo}},\
  }\bibfield  {title} {\enquote {\bibinfo {title} {How to share a quantum
  secret},}\ }\href {\doibase 10.1103/PhysRevLett.83.648} {\bibfield  {journal}
  {\bibinfo  {journal} {Phys. Rev. Lett.}\ }\textbf {\bibinfo {volume} {83}},\
  \bibinfo {pages} {648--651} (\bibinfo {year} {1999})}\BibitemShut {NoStop}%
\bibitem [{\citenamefont {Shor}\ and\ \citenamefont
  {Laflamme}(1997)}]{PhysRevLett.78.1600}%
  \BibitemOpen
  \bibfield  {author} {\bibinfo {author} {\bibfnamefont {P.}~\bibnamefont
  {Shor}}\ and\ \bibinfo {author} {\bibfnamefont {R.}~\bibnamefont
  {Laflamme}},\ }\bibfield  {title} {\enquote {\bibinfo {title} {{Quantum
  Analog of the MacWilliams Identities for Classical Coding Theory}},}\ }\href
  {\doibase 10.1103/PhysRevLett.78.1600} {\bibfield  {journal} {\bibinfo
  {journal} {Phys. Rev. Lett.}\ }\textbf {\bibinfo {volume} {78}},\ \bibinfo
  {pages} {1600--1602} (\bibinfo {year} {1997})}\BibitemShut {NoStop}%
\bibitem [{\citenamefont {Rains}(1998)}]{681316}%
  \BibitemOpen
  \bibfield  {author} {\bibinfo {author} {\bibfnamefont {E.~M.}\ \bibnamefont
  {Rains}},\ }\bibfield  {title} {\enquote {\bibinfo {title} {{Quantum weight
  enumerators}},}\ }\href {\doibase 10.1109/18.681316} {\bibfield  {journal}
  {\bibinfo  {journal} {IEEE Trans. Inf. Theory}\ }\textbf {\bibinfo {volume}
  {44}},\ \bibinfo {pages} {1388--1394} (\bibinfo {year} {1998})}\BibitemShut
  {NoStop}%
\bibitem [{\citenamefont {Huber}\ \emph {et~al.}(2018)\citenamefont {Huber},
  \citenamefont {Eltschka}, \citenamefont {Siewert},\ and\ \citenamefont
  {G{\"u}hne}}]{1751-8121-51-17-175301}%
  \BibitemOpen
  \bibfield  {author} {\bibinfo {author} {\bibfnamefont {F.}~\bibnamefont
  {Huber}}, \bibinfo {author} {\bibfnamefont {C.}~\bibnamefont {Eltschka}},
  \bibinfo {author} {\bibfnamefont {J.}~\bibnamefont {Siewert}}, \ and\
  \bibinfo {author} {\bibfnamefont {O.}~\bibnamefont {G{\"u}hne}},\ }\bibfield
  {title} {\enquote {\bibinfo {title} {{Bounds on absolutely maximally
  entangled states from shadow inequalities, and the quantum MacWilliams
  identity}},}\ }\href {http://doi.org/10.1088/1751-8121/aaade5}
  {\bibfield  {journal} {\bibinfo  {journal} {J. Phys. A: Math. Theor.}\
  }\textbf {\bibinfo {volume} {51}},\ \bibinfo {pages} {175301} (\bibinfo
  {year} {2018})}\BibitemShut {NoStop}%
\bibitem [{\citenamefont {Ashikhmin}\ and\ \citenamefont
  {Barg}(1999)}]{AshikhminBarg99}%
  \BibitemOpen
  \bibfield  {author} {\bibinfo {author} {\bibfnamefont {A.}~\bibnamefont
  {Ashikhmin}}\ and\ \bibinfo {author} {\bibfnamefont {A.}~\bibnamefont
  {Barg}},\ }\bibfield  {title} {\enquote {\bibinfo {title} {Binomial moments
  of the distance distribution: bounds and applications},}\ }\href {\doibase 10.1109/18.748994} {\bibfield  {journal} {\bibinfo  {journal} {IEEE
  Transactions on Information Theory}\ }\textbf {\bibinfo {volume} {45}},\
  \bibinfo {pages} {438--452} (\bibinfo {year} {1999})}\BibitemShut {NoStop}%
\bibitem [{\citenamefont {Gottesman}(1997)}]{quant-ph/9705052v1}%
  \BibitemOpen
  \bibfield  {author} {\bibinfo {author} {\bibfnamefont {D.}~\bibnamefont
  {Gottesman}},\ }\emph {\bibinfo {title} {{Stabilizer Codes and Quantum Error
  Correction}}},\ \href@noop {} {Ph.D. thesis},\ \bibinfo  {school} {Caltech}
  (\bibinfo {year} {1997}),\ \Eprint {http://arxiv.org/abs/quant-ph/9705052}
  {arXiv:quant-ph/9705052}\BibitemShut {NoStop}%
\bibitem [{\citenamefont {Alsina}\ and\ \citenamefont
  {Razavi}(2019)}]{alsina2019absolutely}%
  \BibitemOpen
  \bibfield  {author} {\bibinfo {author} {\bibfnamefont {D.}~\bibnamefont
  {Alsina}}\ and\ \bibinfo {author} {\bibfnamefont {M.}~\bibnamefont
  {Razavi}},\ }\href@noop {} {\enquote {\bibinfo {title} {Absolutely maximally
  entangled states, quantum maximum distance separable codes, and quantum
  repeaters},}\ } (\bibinfo {year} {2019}),\ \Eprint
  {http://arxiv.org/abs/1907.11253} {arXiv:1907.11253}\BibitemShut {NoStop}%
\bibitem [{\citenamefont {Brun}\ and\ \citenamefont
  {Lidar}(2013)}]{Brun_Lidar_QECC}%
  \BibitemOpen
  \bibfield  {author} {\bibinfo {author} {\bibfnamefont {T.~A.}\ \bibnamefont
  {Brun}}\ and\ \bibinfo {author} {\bibfnamefont {D.~E.}\ \bibnamefont
  {Lidar}},\ }\href {\doibase 10.1017/CBO9781139034807} {\emph {\bibinfo
  {title} {{Quantum Error Correction}}}}\ (\bibinfo  {publisher} {Cambridge
  University Press},\ \bibinfo {year} {2013})\BibitemShut {NoStop}%
\bibitem [{\citenamefont
  {Winter}(2019)}]{Winter_private-comm_QMDS-purity_2019}%
  \BibitemOpen
  \bibfield  {author} {\bibinfo {author} {\bibfnamefont {A.}~\bibnamefont
  {Winter}},\ }\href@noop {} {}\bibinfo {howpublished} {private communication}
  (\bibinfo {year} {2019})\BibitemShut {NoStop}%
\bibitem [{\citenamefont {Huffman}\ and\ \citenamefont
  {Pless}(2003)}]{HuffmanPless2003}%
  \BibitemOpen
  \bibfield  {author} {\bibinfo {author} {\bibfnamefont {W.~C.}\ \bibnamefont
  {Huffman}}\ and\ \bibinfo {author} {\bibfnamefont {V.}~\bibnamefont
  {Pless}},\ }\href {\doibase 10.1017/CBO9780511807077} {\emph {\bibinfo
  {title} {{Fundamentals of Error-Correcting Codes}}}}\ (\bibinfo  {publisher}
  {Cambridge University Press},\ \bibinfo {year} {2003})\BibitemShut {NoStop}%
\bibitem [{\citenamefont {MacWilliams}\ and\ \citenamefont
  {Sloane}(1981)}]{MacWilliams1981}%
  \BibitemOpen
  \bibfield  {author} {\bibinfo {author} {\bibfnamefont {F.~J.}\ \bibnamefont
  {MacWilliams}}\ and\ \bibinfo {author} {\bibfnamefont {N.~J.~A.}\
  \bibnamefont {Sloane}},\ }\href@noop {} {\emph {\bibinfo {title} {{The Theory
  of Error-Correcting Codes}}}}\ (\bibinfo  {publisher} {North Holland},\
  \bibinfo {year} {1981})\BibitemShut {NoStop}%
\bibitem [{\citenamefont {Stanley}(2011)}]{Stanley:2011:ECV:2124415}%
  \BibitemOpen
  \bibfield  {author} {\bibinfo {author} {\bibfnamefont {R.~P.}\ \bibnamefont
  {Stanley}},\ }\href {\doibase 10.1017/CBO9780511609589} {\emph {\bibinfo
  {title} {Enumerative Combinatorics: Volume 1}}},\ \bibinfo {edition} {2nd}\
  ed.\ (\bibinfo  {publisher} {Cambridge University Press},\ \bibinfo {year}
  {2011})\BibitemShut {NoStop}%
\bibitem [{\citenamefont {Scott}(2004)}]{PhysRevA.69.052330}%
  \BibitemOpen
  \bibfield  {author} {\bibinfo {author} {\bibfnamefont {A.~J.}\ \bibnamefont
  {Scott}},\ }\bibfield  {title} {\enquote {\bibinfo {title} {{Multipartite
  entanglement, quan\-tum-error-correcting codes, and entangling power of quantum
  evolutions}},}\ }\href {\doibase 10.1103/PhysRevA.69.052330} {\bibfield
  {journal} {\bibinfo  {journal} {Phys. Rev. A}\ }\textbf {\bibinfo {volume}
  {69}},\ \bibinfo {pages} {052330} (\bibinfo {year} {2004})}\BibitemShut
  {NoStop}%
\bibitem [{\citenamefont {Ketkar}\ \emph {et~al.}(2006)\citenamefont {Ketkar},
  \citenamefont {Klappenecker}, \citenamefont {Kumar},\ and\ \citenamefont
  {Sarvepalli}}]{1715533}%
  \BibitemOpen
  \bibfield  {author} {\bibinfo {author} {\bibfnamefont {A.}~\bibnamefont
  {Ketkar}}, \bibinfo {author} {\bibfnamefont {A.}~\bibnamefont
  {Klappenecker}}, \bibinfo {author} {\bibfnamefont {S.}~\bibnamefont {Kumar}},
  \ and\ \bibinfo {author} {\bibfnamefont {P.~K.}\ \bibnamefont {Sarvepalli}},\
  }\bibfield  {title} {\enquote {\bibinfo {title} {{Nonbinary Stabilizer Codes
  Over Finite Fields}},}\ }\href {\doibase 10.1109/TIT.2006.883612} {\bibfield
  {journal} {\bibinfo  {journal} {IEEE Trans. Inf. Theory}\ }\textbf {\bibinfo
  {volume} {52}},\ \bibinfo {pages} {4892--4914} (\bibinfo {year}
  {2006})}\BibitemShut {NoStop}%
\bibitem [{\citenamefont {Helwig}\ \emph {et~al.}(2012)\citenamefont {Helwig},
  \citenamefont {Cui}, \citenamefont {Latorre}, \citenamefont {Riera},\ and\
  \citenamefont {Lo}}]{PhysRevA.86.052335}%
  \BibitemOpen
  \bibfield  {author} {\bibinfo {author} {\bibfnamefont {W.}~\bibnamefont
  {Helwig}}, \bibinfo {author} {\bibfnamefont {W.}~\bibnamefont {Cui}},
  \bibinfo {author} {\bibfnamefont {J.~I.}\ \bibnamefont {Latorre}}, \bibinfo
  {author} {\bibfnamefont {A.}~\bibnamefont {Riera}}, \ and\ \bibinfo {author}
  {\bibfnamefont {H.-K.}\ \bibnamefont {Lo}},\ }\bibfield  {title} {\enquote
  {\bibinfo {title} {{Absolute maximal entanglement and quantum secret
  sharing}},}\ }\href {\doibase 10.1103/PhysRevA.86.052335} {\bibfield
  {journal} {\bibinfo  {journal} {Phys. Rev. A}\ }\textbf {\bibinfo {volume}
  {86}},\ \bibinfo {pages} {052335} (\bibinfo {year} {2012})}\BibitemShut
  {NoStop}%
\bibitem [{\citenamefont {Huber}\ \emph {et~al.}(2017)\citenamefont {Huber},
  \citenamefont {G{\"u}hne},\ and\ \citenamefont
  {Siewert}}]{PhysRevLett.118.200502}%
  \BibitemOpen
  \bibfield  {author} {\bibinfo {author} {\bibfnamefont {F.}~\bibnamefont
  {Huber}}, \bibinfo {author} {\bibfnamefont {O.}~\bibnamefont {G{\"u}hne}}, \
  and\ \bibinfo {author} {\bibfnamefont {J.}~\bibnamefont {Siewert}},\
  }\bibfield  {title} {\enquote {\bibinfo {title} {{Absolutely Maximally
  Entangled States of Seven Qubits Do Not Exist}},}\ }\href {\doibase 10.1103/PhysRevLett.118.200502} {\bibfield  {journal} {\bibinfo  {journal}
  {Phys. Rev. Lett.}\ }\textbf {\bibinfo {volume} {118}},\ \bibinfo {pages}
  {200502} (\bibinfo {year} {2017})}\BibitemShut {NoStop}%
\bibitem [{\citenamefont {Huber}\ and\ \citenamefont
  {Wyderka}(2020)}]{HuberWyderka:ametable}%
  \BibitemOpen
  \bibfield  {author} {\bibinfo {author} {\bibfnamefont {F.}~\bibnamefont
  {Huber}}\ and\ \bibinfo {author} {\bibfnamefont {N.}~\bibnamefont
  {Wyderka}},\ }\href@noop {} {\enquote {\bibinfo {title} {{Table of AME
  states}},}\ }\bibinfo {howpublished} {available online at
  \url{http://www.tp.nt.uni-siegen.de/+fhuber/ame.html}} (\bibinfo {year}
  {2020})\BibitemShut {NoStop}%
\bibitem [{\citenamefont {Rains}(2000)}]{817508}%
  \BibitemOpen
  \bibfield  {author} {\bibinfo {author} {\bibfnamefont {E.~M.}\ \bibnamefont
  {Rains}},\ }\bibfield  {title} {\enquote {\bibinfo {title} {{Polynomial
  invariants of quantum codes}},}\ }\href {\doibase 10.1109/18.817508}
  {\bibfield  {journal} {\bibinfo  {journal} {IEEE Trans. Inf. Theory}\
  }\textbf {\bibinfo {volume} {46}},\ \bibinfo {pages} {54--59} (\bibinfo
  {year} {2000})}\BibitemShut {NoStop}%
\bibitem [{\citenamefont {Rains}(1999{\natexlab{b}})}]{796376}%
  \BibitemOpen
  \bibfield  {author} {\bibinfo {author} {\bibfnamefont {E.~M.}\ \bibnamefont
  {Rains}},\ }\bibfield  {title} {\enquote {\bibinfo {title} {{Quantum shadow
  enumerators}},}\ }\href {\doibase 10.1109/18.796376} {\bibfield  {journal}
  {\bibinfo  {journal} {IEEE Trans. Inf. Theory}\ }\textbf {\bibinfo {volume}
  {45}},\ \bibinfo {pages} {2361--2366} (\bibinfo {year}
  {1999}{\natexlab{b}})}\BibitemShut {NoStop}%
\bibitem [{\citenamefont {Nebe}\ \emph {et~al.}(2006)\citenamefont {Nebe},
  \citenamefont {Rains},\ and\ \citenamefont {Sloane}}]{Nebe2006}%
  \BibitemOpen
  \bibfield  {author} {\bibinfo {author} {\bibfnamefont {G.}~\bibnamefont
  {Nebe}}, \bibinfo {author} {\bibfnamefont {E.~M.}\ \bibnamefont {Rains}}, \
  and\ \bibinfo {author} {\bibfnamefont {N.~J.~A.}\ \bibnamefont {Sloane}},\
  }\href {\doibase 10.1007/3-540-30731-1} {\emph {\bibinfo {title} {{Self-Dual
  Codes and Invariant Theory}}}}\ (\bibinfo  {publisher} {Springer
  Berlin-Heidelberg},\ \bibinfo {year} {2006})\BibitemShut {NoStop}%
\bibitem [{\citenamefont {Rains}(1999{\natexlab{c}})}]{746807}%
  \BibitemOpen
  \bibfield  {author} {\bibinfo {author} {\bibfnamefont {E.~M.}\ \bibnamefont
  {Rains}},\ }\bibfield  {title} {\enquote {\bibinfo {title} {{Quantum codes of
  minimum distance two}},}\ }\href {\doibase 10.1109/18.746807} {\bibfield
  {journal} {\bibinfo  {journal} {IEEE Trans. Inf. Theory}\ }\textbf {\bibinfo
  {volume} {45}},\ \bibinfo {pages} {266--271} (\bibinfo {year}
  {1999}{\natexlab{c}})}\BibitemShut {NoStop}%
\bibitem [{\citenamefont {Horodecki}\ \emph {et~al.}(2020)\citenamefont
  {Horodecki}, \citenamefont {Rudnicki},\ and\ \citenamefont
  {{\.Z}yczkowski}}]{KCIK2019}%
  \BibitemOpen
  \bibfield  {author} {\bibinfo {author} {\bibfnamefont {P.}~\bibnamefont
  {Horodecki}}, \bibinfo {author} {\bibfnamefont {L.}~\bibnamefont {Rudnicki}},
  \ and\ \bibinfo {author} {\bibfnamefont {K.}~\bibnamefont {{\.Z}yczkowski}},\
  }\href@noop {} {\enquote {\bibinfo {title} {Five open problems in quantum
  information},}\ } (\bibinfo {year} {2020}),\ \Eprint
  {http://arxiv.org/abs/2002.03233} {arXiv:2002.03233}\BibitemShut {NoStop}%
\bibitem [{\citenamefont {Ashikhmin}(1997)}]{Ashikhmin1997}%
  \BibitemOpen
  \bibfield  {author} {\bibinfo {author} {\bibfnamefont {A.}~\bibnamefont
  {Ashikhmin}},\ }\href@noop {} {\enquote {\bibinfo {title} {{Remarks on Bounds
  for Quantum Codes}},}\ } (\bibinfo {year} {1997}),\ \Eprint
  {http://arxiv.org/abs/quant-ph/9705037} {arXiv:quant-ph/9705037}\BibitemShut
  {NoStop}%
\bibitem [{\citenamefont {Ashikhmin}\ and\ \citenamefont
  {Litsyu}(1999)}]{AshikhminLitsyn99}%
  \BibitemOpen
  \bibfield  {author} {\bibinfo {author} {\bibfnamefont {A.}~\bibnamefont
  {Ashikhmin}}\ and\ \bibinfo {author} {\bibfnamefont {S.}~\bibnamefont
  {Litsyu}},\ }\bibfield  {title} {\enquote {\bibinfo {title} {Upper bounds on
  the size of quantum codes},}\ }\href {\doibase 10.1109/18.761270} {\bibfield
  {journal} {\bibinfo  {journal} {IEEE Transactions on Information Theory}\
  }\textbf {\bibinfo {volume} {45}},\ \bibinfo {pages} {1206--1215} (\bibinfo
  {year} {1999})}\BibitemShut {NoStop}%
\bibitem [{\citenamefont {Nielsen}\ and\ \citenamefont
  {Chuang}(2011)}]{NielsenChuang2011}%
  \BibitemOpen
  \bibfield  {author} {\bibinfo {author} {\bibfnamefont {M.~A.}\ \bibnamefont
  {Nielsen}}\ and\ \bibinfo {author} {\bibfnamefont {I.~L.}\ \bibnamefont
  {Chuang}},\ }\href {\doibase 10.1017/CBO9780511976667} {\emph {\bibinfo
  {title} {{Quantum Computation and Quantum Information: 10th Anniversary
  Edition}}}}\ (\bibinfo  {publisher} {Cambridge University Press},\ \bibinfo
  {year} {2011})\BibitemShut {NoStop}%
\bibitem [{\citenamefont {Müller-Hermes}\ \emph {et~al.}(2016)\citenamefont
  {Müller-Hermes}, \citenamefont {Stilck~França},\ and\ \citenamefont
  {Wolf}}]{doi:10.1063/1.4939560}%
  \BibitemOpen
  \bibfield  {author} {\bibinfo {author} {\bibfnamefont {A.}~\bibnamefont
  {Müller-Hermes}}, \bibinfo {author} {\bibfnamefont {D.}~\bibnamefont
  {Stilck~França}}, \ and\ \bibinfo {author} {\bibfnamefont {M.~M.}\
  \bibnamefont {Wolf}},\ }\bibfield  {title} {\enquote {\bibinfo {title}
  {Relative entropy convergence for depolarizing channels},}\ }\href {\doibase 10.1063/1.4939560} {\bibfield  {journal} {\bibinfo  {journal} {J. Math.
  Phys.}\ }\textbf {\bibinfo {volume} {57}},\ \bibinfo {pages} {022202}
  (\bibinfo {year} {2016})}\BibitemShut {NoStop}%
\bibitem [{\citenamefont {Klappenecker}\ and\ \citenamefont
  {Rötteler}(2002)}]{1019848}%
  \BibitemOpen
  \bibfield  {author} {\bibinfo {author} {\bibfnamefont {A.}~\bibnamefont
  {Klappenecker}}\ and\ \bibinfo {author} {\bibfnamefont {M.}~\bibnamefont
  {Rötteler}},\ }\bibfield  {title} {\enquote {\bibinfo {title} {{Beyond
  stabilizer codes. I. Nice error bases}},}\ }\href {\doibase 10.1109/TIT.2002.800471} {\bibfield  {journal} {\bibinfo  {journal} {IEEE
  Trans. Inf. Theory}\ }\textbf {\bibinfo {volume} {48}},\ \bibinfo {pages}
  {2392--2395} (\bibinfo {year} {2002})}\BibitemShut {NoStop}%
\bibitem [{\citenamefont {{Rötteler}}\ \emph {et~al.}(2004)\citenamefont
  {{Rötteler}}, \citenamefont {{Grassl}},\ and\ \citenamefont
  {{Beth}}}]{1365393}%
  \BibitemOpen
  \bibfield  {author} {\bibinfo {author} {\bibfnamefont {M.}~\bibnamefont
  {{Rötteler}}}, \bibinfo {author} {\bibfnamefont {M.}~\bibnamefont
  {{Grassl}}}, \ and\ \bibinfo {author} {\bibfnamefont {T.}~\bibnamefont
  {{Beth}}},\ }\bibfield  {title} {\enquote {\bibinfo {title} {On quantum {MDS}
  codes},}\ }in\ \href {\doibase 10.1109/ISIT.2004.1365393} {\emph {\bibinfo
  {booktitle} {Proceedings, Int. Symp. Inf. Theory (ISIT 2004)}}}\ (\bibinfo
  {year} {2004})\ p.\ \bibinfo {pages} {356}\BibitemShut {NoStop}%
\bibitem [{\citenamefont {Sarvepalli}\ and\ \citenamefont
  {Klappenecker}(2005)}]{1523494}%
  \BibitemOpen
  \bibfield  {author} {\bibinfo {author} {\bibfnamefont {P.~K.}\ \bibnamefont
  {Sarvepalli}}\ and\ \bibinfo {author} {\bibfnamefont {A.}~\bibnamefont
  {Klappenecker}},\ }\bibfield  {title} {\enquote {\bibinfo {title} {{Nonbinary
  quantum Reed-Muller codes}},}\ }in\ \href {\doibase 10.1109/ISIT.2005.1523494} {\emph {\bibinfo {booktitle} {{Proceedings. Int.
  Symp. Inf. Theory, (ISIT 2005)}}}}\ (\bibinfo {year} {2005})\ pp.\ \bibinfo
  {pages} {1023--1027}\BibitemShut {NoStop}%
\bibitem [{\citenamefont {{Jin}}\ \emph {et~al.}(2010)\citenamefont {{Jin}},
  \citenamefont {{Ling}}, \citenamefont {{Luo}},\ and\ \citenamefont
  {{Xing}}}]{5550401}%
  \BibitemOpen
  \bibfield  {author} {\bibinfo {author} {\bibfnamefont {L.}~\bibnamefont
  {{Jin}}}, \bibinfo {author} {\bibfnamefont {S.}~\bibnamefont {{Ling}}},
  \bibinfo {author} {\bibfnamefont {J.}~\bibnamefont {{Luo}}}, \ and\ \bibinfo
  {author} {\bibfnamefont {C.}~\bibnamefont {{Xing}}},\ }\bibfield  {title}
  {\enquote {\bibinfo {title} {Application of classical {Hermitian}
  self-orthogonal {MDS} codes to quantum {MDS} codes},}\ }\href {\doibase 10.1109/TIT.2010.2054174} {\bibfield  {journal} {\bibinfo  {journal} {IEEE
  Trans. Inf. Theory}\ }\textbf {\bibinfo {volume} {56}},\ \bibinfo {pages}
  {4735--4740} (\bibinfo {year} {2010})}\BibitemShut {NoStop}%
\bibitem [{\citenamefont {Li}\ and\ \citenamefont
  {Xu}(2010)}]{PhysRevA.82.052316}%
  \BibitemOpen
  \bibfield  {author} {\bibinfo {author} {\bibfnamefont {R.}~\bibnamefont
  {Li}}\ and\ \bibinfo {author} {\bibfnamefont {Z.}~\bibnamefont {Xu}},\
  }\bibfield  {title} {\enquote {\bibinfo {title} {{Construction of
  $[\![n,n\ensuremath{-}4,3]\!]{}_{q}$ quantum codes for odd prime power $q$}},}\
  }\href {\doibase 10.1103/PhysRevA.82.052316} {\bibfield  {journal} {\bibinfo
  {journal} {Phys. Rev. A}\ }\textbf {\bibinfo {volume} {82}},\ \bibinfo
  {pages} {052316} (\bibinfo {year} {2010})}\BibitemShut {NoStop}%
\bibitem [{\citenamefont {Ball}(2019)}]{Ball2019}%
  \BibitemOpen
  \bibfield  {author} {\bibinfo {author} {\bibfnamefont {S.}~\bibnamefont
  {Ball}},\ }\href@noop {} {\enquote {\bibinfo {title} {{Some constructions of
  quantum MDS codes}},}\ } (\bibinfo {year} {2019}),\ \Eprint
  {http://arxiv.org/abs/1907.04391} {arXiv:1907.04391}\BibitemShut {NoStop}%
\bibitem [{\citenamefont {Glynn}(1986)}]{GLYNN198643}%
  \BibitemOpen
  \bibfield  {author} {\bibinfo {author} {\bibfnamefont {D.~G.}\ \bibnamefont
  {Glynn}},\ }\bibfield  {title} {\enquote {\bibinfo {title} {The non-classical
  $10$-arc of {$PG(4, 9)$}},}\ }\href {\doibase 10.1016/0012-365X(86)90067-1}
  {\bibfield  {journal} {\bibinfo  {journal} {Discrete Math.}\ }\textbf
  {\bibinfo {volume} {59}},\ \bibinfo {pages} {43--51} (\bibinfo {year}
  {1986})}\BibitemShut {NoStop}%
\bibitem [{\citenamefont {{Gulliver}}\ \emph {et~al.}(2008)\citenamefont
  {{Gulliver}}, \citenamefont {{Kim}},\ and\ \citenamefont {{Lee}}}]{4608969}%
  \BibitemOpen
  \bibfield  {author} {\bibinfo {author} {\bibfnamefont {T.~A.}\ \bibnamefont
  {{Gulliver}}}, \bibinfo {author} {\bibfnamefont {J.}~\bibnamefont {{Kim}}}, \
  and\ \bibinfo {author} {\bibfnamefont {Y.}~\bibnamefont {{Lee}}},\ }\bibfield
   {title} {\enquote {\bibinfo {title} {New {MDS} or near-{MDS} self-dual
  codes},}\ }\href {\doibase 10.1109/TIT.2008.928297} {\bibfield  {journal}
  {\bibinfo  {journal} {IEEE Trans. Inf. Theory}\ }\textbf {\bibinfo {volume}
  {54}},\ \bibinfo {pages} {4354--4360} (\bibinfo {year} {2008})}\BibitemShut
  {NoStop}%
\bibitem [{\citenamefont {Kim}\ and\ \citenamefont {Lee}(2004)}]{KimLee04}%
  \BibitemOpen
  \bibfield  {author} {\bibinfo {author} {\bibfnamefont {J.-L.}\ \bibnamefont
  {Kim}}\ and\ \bibinfo {author} {\bibfnamefont {Y.}~\bibnamefont {Lee}},\
  }\bibfield  {title} {\enquote {\bibinfo {title} {Euclidean and {Hermitian}
  self-dual {MDS} codes over large finite fields},}\ }\href {\doibase 10.1016/j.jcta.2003.10.003} {\bibfield  {journal} {\bibinfo  {journal} {J.
  Comb. Theory A}\ }\textbf {\bibinfo {volume} {105}},\ \bibinfo {pages}
  {79–95} (\bibinfo {year} {2004})}\BibitemShut {NoStop}%
\end{thebibliography}%
\onecolumngrid %balance the columns on the final page
\end{document}